\documentclass[10pt,a4paper,twoside,reqno]{amsart} 
\usepackage[utf8]{inputenc}
\usepackage[english]{babel}
\usepackage{graphicx}
\usepackage[a4paper,width=150mm,top=25mm,bottom=25mm,bindingoffset=6mm]{geometry}
\usepackage{fancyhdr}

\usepackage[bookmarks=false]{hyperref}
\usepackage{frontespizio}

\usepackage{float}
\usepackage{amssymb}
\usepackage{latexsym}
\usepackage{amsthm}
\usepackage{amsmath}
\usepackage{amssymb}
\usepackage{dsfont}
\usepackage{mathrsfs}
 \usepackage{hyperref}
 \usepackage{multirow}
 \usepackage{booktabs}

\usepackage{tikz}
\usepackage{tikz-qtree}

\usepackage[skip=2pt]{caption}
\usepackage{tensor}

\makeatletter
\renewcommand{\tocsection}[3]{%
  \indentlabel{\@ifnotempty{#2}{\bfseries\ignorespaces#1 #2\quad}}\bfseries#3}
\renewcommand{\tocsubsection}[3]{%
  \indentlabel{\@ifnotempty{#2}{\ignorespaces#1 #2\quad}}#3}

\newcommand\@dotsep{4.5}
\def\@tocline#1#2#3#4#5#6#7{\relax
  \ifnum #1>\c@tocdepth 
  \else
    \par \addpenalty\@secpenalty\addvspace{#2}%
    \begingroup \hyphenpenalty\@M
    \@ifempty{#4}{%
      \@tempdima\csname r@tocindent\number#1\endcsname\relax
    }{%
      \@tempdima#4\relax
    }%
    \parindent\z@ \leftskip#3\relax \advance\leftskip\@tempdima\relax
    \rightskip\@pnumwidth plus1em \parfillskip-\@pnumwidth
    #5\leavevmode\hskip-\@tempdima{#6}\nobreak
    \leaders\hbox{$\m@th\mkern \@dotsep mu\hbox{.}\mkern \@dotsep mu$}\hfill
    \nobreak
    \hbox to\@pnumwidth{\@tocpagenum{\ifnum#1=1\bfseries\fi#7}}\par
    \nobreak
    \endgroup
  \fi}
\AtBeginDocument{%
\expandafter\renewcommand\csname r@tocindent0\endcsname{0pt}
}
\makeatother

\usepackage{lipsum}

\DeclareUnicodeCharacter{2009}{\,}

\newtheorem{thm}{Theorem}[section]
\newtheorem{defn}[thm]{Definition}
\newtheorem{lem}[thm]{Lemma}
\newtheorem{cor}[thm]{Corollary}

\newtheorem{prop}[thm]{Proposition}
\newtheorem{oss}[thm]{Remark}

\usepackage{amsaddr}

\usepackage{abstract}

\begin{document}

\title[Nonexistence of closed timelike geodesics in Kerr spacetimes]{Nonexistence of closed timelike geodesics\\ in Kerr spacetimes}

\author{Giulio Sanzeni}

\address{Ruhr-Universit\"at Bochum,  Fakult\"at f\"ur Mathematik,  Universit\"atsstra\ss e 150,  44801,  Bochum, Germany}
\email{ giulio.sanzeni@rub.de}

\maketitle

\vspace{0.5cm}

 \begin{abstract}
The Kerr-star spacetime is the extension over the horizons and in the negative radial region of the Kerr spacetime.  Despite the presence of closed timelike curves below the inner horizon,  we prove that the timelike geodesics cannot be closed in the Kerr-star spacetime.  Since the existence of closed null geodesics was ruled out by the author in [G.  Sanzeni,\cite{sanzeni2024non} (2024)],   this result shows the absence of closed causal geodesics in the Kerr-star spacetime.
\end{abstract} 

\vspace{0.5cm}


\setcounter{tocdepth}{2}
\setcounter{tocdepth}{3}

\section{\textbf{Introduction}}

\subsection{The Kerr solution and its chronology violations}

The Kerr spacetime is a stationary,  axisymmetric and asymptotically flat black hole solution of Einstein’s vacuum field equations found by R. P.  Kerr \cite{Kerr-paper}.  This spacetime depends on a \textit{mass parameter}  $M$ and a \textit{rotation parameter} $a$ (angular momentum per unit mass).  The static spherically symmetric Schwarzschild solution \cite{Schw_paper} is obtained from the Kerr solution in the limit case $a=0$. The slowly rotating ($|a|<M$) Kerr spacetime have two horizons,  an outer \textit{event horizon} and an inner \textit{causality horizon}.  If the Kerr spacetime is analytically extended over the horizons and in the negative radial region \cite{Boyer-Lindquist_paper,  KBH_book},  from now on called the Kerr-star spacetime,  through every point below the causality horizon there exists a closed timelike curve \cite{Carter_causality}.  In this paper,  we prove that despite chronology violations,  the timelike geodesics cannot be closed in the Kerr-star spacetime.  This work follows the strategy adopted in \cite{sanzeni2024non} in which we proved the absence of closed null geodesics.  Therefore as the G\"odel spacetime \cite{Godel},  the Kerr-star spacetime is not causal but it does not contain closed causal geodesics,  see \cite{Kundt,  Chandr_Wright,  Nolan_godel}.


\subsection{Result}\label{subsection Result}
Consider a spacetime $\big(\mathcal{M},\mathbf{g}\big)$,  \textit{i.e.} a time-oriented connected Lorentzian manifold,  and a geodesic curve $\gamma:I=[a,b]\rightarrow \mathcal{M}$.   $\gamma$ is called \textit{closed geodesic}  if $\gamma(a)=\gamma(b)$ and $\gamma'(a)=\lambda\gamma'(b)\neq 0$,  for some real number $\lambda\neq 0$.  If $\gamma$ is timelike,  then $\lambda=1$.
The purpose of this paper is to prove the nonexistence of closed timelike geodesics in the Kerr-star extension of the slowly rotating ($|a|<M$) Kerr black hole,  described in detail in \S \ref{definiton of Kerr}. 

\begin{thm}\label{main theorem}
Let $K^*$ be the Kerr-star spacetime.  Then there are no closed timelike geodesics in $K^*$.
\end{thm}
The nonexistence of closed null geodesics in the Kerr-star spacetime was proved in the following result.

\begin{thm}[Theorem $1.1$, \cite{sanzeni2024non}]\label{thm no closed null}
Let $K^*$ be the Kerr-star spacetime.  Then there are no closed null geodesics in $K^*$.
\end{thm}

\begin{cor}
Let $K^*$ be the Kerr-star spacetime.  Then there are no closed causal geodesics in $K^*$.
\end{cor}

\subsection{Geodesic motion in Kerr spacetimes}
The Kerr spacetimes are completely integrable systems.  Indeed for any geodesic there exist four independent constants of motion: the \textit{energy} (associated to a timelike Killing vector field),  the \textit{angular momentum} (associated to a spacelike Killing vector field),  the \textit{Lorentzian energy} (the causal character of the geodesic) and the \textit{Carter constant} (associated to a Killing $2$-tensor) \cite{Carter_causality}.  Therefore one can study the geodesic motion solving a system of four coupled first-order differential equations \cite{Carter_causality,  KBH_book} .  Geodesics restricted on submanifolds were firstly studied.  Boyer and Price \cite{Boyer_Price_1965},  then Boyer and Lindquist \cite{Boyer-Lindquist_paper},  and hence de Felice \cite{deFelice_1968} considered geodesic motion in the equatorial hyperplane $Eq=\{\theta=\pi/2\}$.  Geodesics in the axis of symmetry $A=\{\theta=0,\pi\}$ of the black hole were analysed by Carter \cite{Carter_1966_Axis}.  Wilkins instead studied trapped orbits,  namely geodesics running over a finite radial interval \cite{Wilkins}.  The most exhaustive references about geodesic motion in Kerr spacetimes are the text-books by Chandrasekhar \cite{Chandrasekhar} and O'Neill \cite{KBH_book}.  In this paper,  we first ruled out the existence of closed timelike geodesics strictly contained in $\{0<r<r_{-}\}$ (Prop.  \ref{prop spacelike foliation}),  intersecting the horizons (Prop.  \ref{oss event horizon}) and tangent to the axis (Prop.  \ref{no closed geodesics in axis}).  Starting from \S \ref{steps of other cases},  the remaining timelike geodesics are analyzed.  It turned out that the most difficult ones to investigate are those with non-vanishing energy and negative Carter constant.  Firstly,  we observed that if a geodesic of such kind is closed,  it must have constant $r$-coordinate,  so it must be a spherical geodesic (Prop.  \ref{geod Q<0}).  Secondly,  we obtained a lower bound on the negative constant $r$-coordinate (Prop.  \ref{spherical timelike geodesics have r>-M}).   Finally,  arguing by contradiction we proved that the variation of the $t$-coordinate (see eq.  \eqref{last espression for delta t}) on a full $\theta$-oscillation must be positive (Prop.  \ref{last proposition}) for any spherical timelike geodesic with negative Carter constant.  Therefore it is shown that the timelike geodesics cannot be closed in the Kerr-star spacetime.

\subsection{Organization of the paper}

In \S \ref{definiton of Kerr},  we introduce the Kerr metric and  discuss the definition and properties of the Kerr-star spacetime.  In \S \ref{study of geodesic equations} we recall the set of first order differential equations satisfied by geodesic orbits.  In \S \ref{section: properites of timelike geodesics},  we study the properties of timelike geodesics required to prove the main theorem.  In \S\ref{section: main theorem},  we give the proof of Thm.  \ref{main theorem} split into several cases.  The overall structure of the proof is detailed in \ref{strategy of the proof},  \ref{steps of other cases} and Fig.  \ref{figure steps of proof}.

\section{\textbf{ The Kerr-star spacetime}}\label{definiton of Kerr}

Consider $\mathbb{R}^2\times S^2$ with coordinates $(t,r)\in\mathbb{R}^2$ and $(\theta,\phi)\in S^2$.  Fix two real numbers $a\in\mathbb{R}\setminus \{0\}$,  $M\in \mathbb{R}_{>0}$  and define the functions 
\[
 \rho(r,\theta):= \sqrt{r^2+a^2\cos^2\theta}
\]
and 
\[
\Delta(r):=r^2-2Mr+a^2.
\]

We study the case $|a|<M$ called \textit{slow Kerr},  for which $\Delta(r)$ has two positive roots
\begin{align*}
r_{\pm}=M\pm \sqrt{M^2-a^2}>0
\end{align*}
and define two sets 
\begin{itemize}
\item[(1)] the \textit{horizons} $\mathscr{H}:=\{\Delta(r)=0\}=\{r=r_{\pm}\}:=\mathscr{H}_{-}\,\sqcup \mathscr{H}_{+}$, 
\item[(2)] the \textit{ring singularity} $\Sigma:=\{\rho(r,\theta)=0\}=\{r=0,\;\theta=\pi/2\}$.
\end{itemize}

The {\it Kerr metric}  \cite{Kerr-paper} in {\it Boyer--Lindquist coordinates} is

\begin{align}\label{kerr metric}
\mathbf{g}=-dt\otimes dt + \frac{2Mr}{\rho^2(r,\theta)}(dt-a\sin^2\theta\; d\phi)^2 + \frac{\rho^2(r,\theta)}{\Delta(r)}dr\otimes dr + a^2\sin^4(\theta) d\phi\otimes d\phi  + \rho^2(r,\theta) d\sigma^2  ,
\end{align}
where $d\sigma^2=d\theta\otimes d\theta + \sin^2\theta d\phi\otimes d\phi$ is the $2$-dimensional  (Riemannian) metric of constant unit curvature on the unit sphere $S^2\subset\mathbb{R}^3$ written in spherical coordinates. 

\begin{oss}
The components of $\mathbf{g}$ in Boyer--Lindquist coordinates can be read off the common expression

\begin{align}
\mathbf{g}=-&\bigg(1-\frac{2 M r}{\rho^2(r,\theta)} \bigg)\:dt\otimes dt-\frac{4Mar\sin^2\theta}{\rho^2(r,\theta)}\: dt\otimes d\phi + \nonumber\\+ &\bigg(r^2+a^2+\frac{2Mra^2\sin^2\theta}{\rho^2(r,\theta)} \bigg)\sin^2\theta\:  d\phi\otimes d\phi + \frac{\rho^2(r,\theta)}{\Delta(r)}\: dr\otimes dr + \rho^2(r,\theta)\: d\theta\otimes d\theta.
\end{align}

Nevertheless this last expression does not cover the subsets $\{\theta=0,\pi\}$.

\end{oss}

\begin{lem}
The metric \eqref{kerr metric} is a Lorentzian metric on $\mathbb{R}^2\times S^2\setminus (\Sigma\,\cup \mathscr{H})$.
\end{lem}

The Boyer--Lindquist coordinates or the metric tensor fail on the sets $\mathscr{H}$ and $\Sigma$.  In order to extend the metric tensor to the horizons,  one has to introduce a new set of coordinates.  No change of coordinates can be found in order to extend the metric across the ring singularity.  For a detailed study of the nature of the ring singularity,  see for instance \cite{Chrusc_singularity}.

\begin{defn}
The subsets 
\[
\textrm{I}:=\{r>r_{+}\},\;  \textrm{II}:=\{r_{-}<r<r_{+}\},\; \textrm{III}:=\{r<r_{-}\}\subset \{(t,r)\in\mathbb{R}^2,\; (\theta,\phi)\in S^2 \}\setminus (\Sigma\,\cup \mathscr{H})
\] 
are called the Boyer--Lindquist (BL) blocks.
\end{defn}

\begin{oss}
The BL blocks I,  II and III are the connected components of $\mathbb{R}^2\times S^2\setminus (\Sigma\,\cup \mathscr{H})$.  Each block with the restriction of the metric tensor \eqref{kerr metric} is a connected Lorentzian $4$-manifold.  To get spacetimes,  one has to choose a time orientation on each block.
\end{oss}

\begin{figure}[H]
\centering
\includegraphics[scale=0.4]{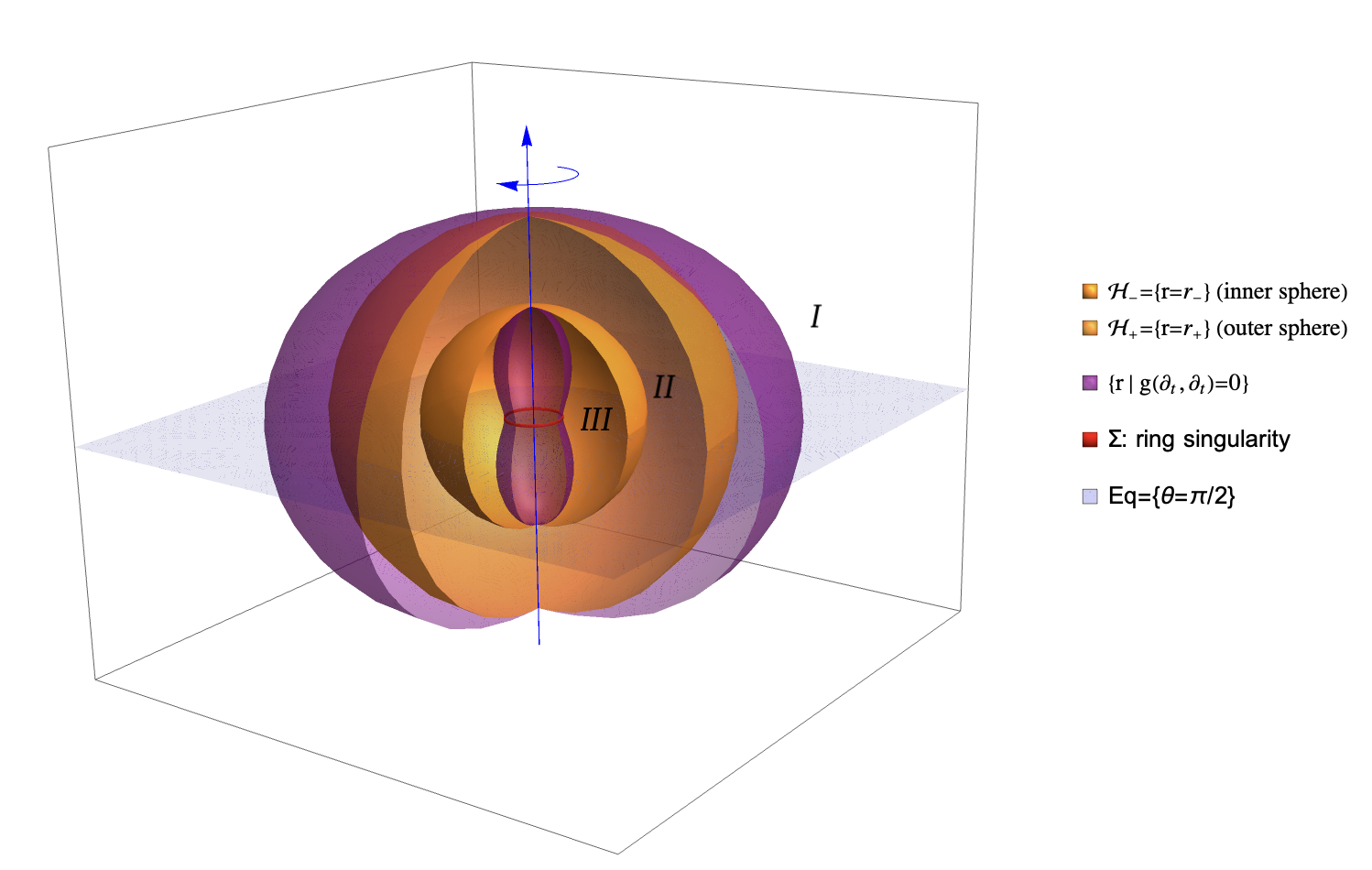} 
\caption{This picture shows a $t$-slice $\{t\}\times\mathbb{R} \times S^2$,  with the radius drawn as $e^r$,  so that $r=-\infty$ is at the center of the figure.  The \textit{Ergoregion} $\{\mathbf{g}(\partial_t,\partial_t)>0\}$ (at fixed time $t$) is  the region between the purple ellipsoids in which $\partial_t$ becomes spacelike.}
\vspace*{-5mm}
\end{figure}

\vspace{1cm}

\subsection{Time orientation of BL blocks }

We define a future time-orientation of block  I using the gradient timelike vector field $-\nabla t$.  Indeed,  the hypersurfaces $\{t=\textrm{const}\}$ are spacelike in block I.  Notice that the coordinate vector field $\partial_t$ is timelike future-directed for $r\gg r_+$ on block I, since $\mathbf{g}(-\nabla t,\partial_t)=-1$.\\

We define a time-orientation of block II by declaring the vector field $-\partial_r$,  which is timelike in II,  to be future-oriented.\\

We define a time-orientation of block III by declaring the vector field $V:=(r^2+a^2)\partial_t + a\partial_\phi$,  which is timelike in III,  to be future-oriented.\\

With this choice of time-orientations,  each block is a \textit{spacetime},  \textit{i.e.} a connected time-oriented Lorentzian $4$-manifold.

\subsection{Kerr spacetimes}

\begin{defn}\label{definition kerr spacetime}
A {\it Kerr spacetime} is an analytic spacetime $(\mathcal{M}_{\textit{Kerr}},\mathbf{g})$ such that
\begin{enumerate}
\item[(1)]  there exists a family of open disjoint isometric embeddings $\Phi_i \colon \mathcal{B}_i\hookrightarrow \mathcal{M}_{\textit{Kerr}}$ $(i\in \mathbb{N})$ of BL blocks $(\mathcal{B}_i,
\mathbf{g}|_{\mathcal{B}_i})$ (\textit{i.e.} $\mathbf{g}|_{\mathcal{B}_i}=\Phi^*_i \mathbf{g}|_{\Phi_i(\mathcal{B}_i)}$) such that $\cup_{i\in\mathbb{N}} \Phi_i(\mathcal{B}_i)$ is dense in $\mathcal{M}_{\textit{Kerr}}$;

\item[(2)] there are analytic functions $r$ and $C$ on $\mathcal{M}_{\textit{Kerr}}$  such that their restriction on each $\Phi_i(\mathcal{B}_i)$ of condition $(1)$ is $\Phi_i$-related to the Boyer--Lindquist functions $r$ and $C=\cos\theta$ on $\mathcal{B}_i$;

\item[(3)]  there is an isometry $\epsilon: \mathcal{M}_{\textit{Kerr}}\rightarrow \mathcal{M}_{\textit{Kerr}}$ called the \textit{equatorial isometry} whose restrictions to each BL block sends $\theta$ to $\pi-\theta$,  leaving the  other coordinates unchanged;

\item[(4)] there are Killing vector fields $\tilde{\partial}_t$ and $\tilde{\partial}_\phi$ on $\mathcal{M}_{\textit{Kerr}}$ that restrict  to the Boyer--Lindquist coordinate vector fields $\partial_t$ and $\partial_\phi$ on each BL block.
\end{enumerate}
\end{defn}

\begin{oss}
With abuse of notation,  we identify each block $\mathcal{B}_i$ with its image via the isometric embedding $\Phi_i(\mathcal{B}_i)\subset \mathcal{M}_{\textit{Kerr}}$.
\end{oss}

\begin{lem}
Each time-oriented BL block is a Kerr spacetime.
\end{lem}

\begin{defn}\label{abstract def axis and Eq}
In a Kerr spacetime $\mathcal{M}_{\textit{Kerr}}$,  on any BL block $\mathcal{B}_i$
\begin{enumerate}
\item the \textit{axis} $A=\{\theta=0,\pi\}$ is the set of zeroes of the Killing vector field $\tilde{\partial}_\phi$ as in $(4)$ of Def. \ref{definition kerr spacetime};
\item the \textit{equatorial hyperplane} $Eq=\{\theta=\pi/2\}$ is the set of fixed points of the equatorial isometry $\epsilon$ as in $(3)$ of Def.  \ref{definition kerr spacetime}.
\end{enumerate}
\end{defn}

\subsection{The Kerr-star spacetime}

\begin{defn}
On each BL block,  we define the \textit{Kerr-star coordinate} functions:
\begin{align}
t^*:=t+\mathcal{T}(r)\in\mathbb{R},\hspace{1cm} \phi^*:=\phi+\mathcal{A}(r)\in S^1,
\end{align}
with $d\mathcal{T}/dr:=(r^2+a^2)/\Delta(r)$ and $d\mathcal{A}/dr:=a/\Delta(r)$.
\end{defn}

\begin{lem}[\cite{KBH_book},  Lemma  $2.5.1$]
For each BL block $B$,  the map $\xi^*=(t^*,r,\theta,\phi^*):B\setminus A\rightarrow \xi^*(B)\subseteq\mathbb{R}^4$ is a coordinate system on $B\setminus A$,  where $A$ is the axis.  We call $\xi^*$ a \textit{Kerr-star} coordinate system.
\end{lem}

Because the  Kerr-star coordinate functions differ from BL coordinates only by additive functions of $r$,  the coordinate vector fields $\partial_t,\partial_\theta,\partial_\phi$ are the same in the two systems,  except that in $K^*$ they extend over the horizons.  However,  the coordinate vector field associated to $r$ does change its form,  and we define $\partial^*_r:=\partial_r-\Delta(r)^{-1}V$,  where $V$ is one of the canonical vector fields defined in Section \ref{study of geodesic equations}.  Note that if we use Kerr-star coordinates,  we get $\mathbf{g}(\partial^*_r,\partial^*_r)=0$,  \textit{i.e.} $\partial^*_r$ is a null vector field of $K^*$,  while in BL coordinates,  $\mathbf{g}(\partial_r,\partial_r)=\rho^2(r,\theta)/\Delta(r),$ which is singular when $\Delta(r)=0$.

\begin{lem}\label{Kerr-star metric}
The Kerr metric,  expressed in Kerr-star coordinates,  takes the form

\begin{align}
\mathbf{g}=&-\bigg(1-\frac{2 M r}{\rho^2(r,\theta)} \bigg) \:dt^*\otimes dt^* -\frac{4Mar\sin^2\theta}{\rho^2(r,\theta)}\: dt^*\otimes d\phi^*\,+\nonumber\\ 
&+ \bigg(r^2+a^2+\frac{2Mra^2\sin^2\theta}{\rho^2(r,\theta)} \bigg)\sin^2\theta\:   d\phi^*\otimes d\phi^* +2\: dt^*\otimes dr\, +\\
&-2a\sin^2\theta\: d\phi^*\otimes dr + \rho^2(r,\theta)\: d\theta\otimes d\theta.  \nonumber
\end{align}
\end{lem}

Now all coefficients in $\mathbf{g}$ are well defined on the horizons $\mathscr{H}=\{\Delta(r)=0\}$,  hence it is a well defined Lorentzian metric on  $\mathbb{R}^2\times S^2\setminus \Sigma$ and constitutes an analytic extension of \eqref{kerr metric} over $\mathscr{H}$.

\begin{defn}\label{Kerr-star spacetime}
The \textit{Kerr-star spacetime} is a Kerr spacetime as defined in \ref{definition kerr spacetime} given by the tuple $(K^*,\mathbf{g},o)$ with $K^*=\{(t^*,r)\in\mathbb{R}^2,\, (\theta,\phi^*)\in S^2\}\setminus\Sigma$,  $\mathbf{g}$ as in Lemma \ref{Kerr-star metric}   (extended over the axis) and $o$ is the future time-orientation induced by the null vector field $-\partial^*_r.$
\end{defn}

\begin{oss}
Note that the time-orientations on individual BL blocks agree with the ones defined for the Kerr-star spacetime:  $\mathbf{g}(-\partial^*_r,\partial_t)=-1<0$ on I,   $\mathbf{g}(-\partial^*_r,-\partial_r)=\mathbf{g}(\partial_r,\partial_r)=\rho^2(r,\theta)/\Delta(r)<0$ on II and $\mathbf{g}(-\partial^*_r,V)=\frac{1}{\Delta(r)}\mathbf{g}(V,V)=-\rho^2(r,\theta)<0$ on III.
\end{oss}

\subsection{Totally geodesic submanifolds of the Kerr-star spacetime}

\begin{lem}[See p. $68$ of \cite{KBH_book}]\label{A and Eq closed totally geod subman}
Let $K^*$ be the Kerr-star spacetime as in Def.  \ref{Kerr-star spacetime}.  The axis $A$ and the equatorial hyperplane $Eq$ of $K^*$ are closed totally geodesic submanifolds of $K^*$.
\end{lem}

\begin{prop}\cite{KBH_book}\label{H is closed totally geod}
Let $K^*$ be the Kerr-star spacetime.  Then the horizon $\mathscr{H}$ is a closed totally geodesic null hypersurface,  with future hemicone on the $-\partial^*_r$ side.  Moreover,  the restriction of $V:=(r^2+a^2)\partial_t+a\partial_\phi$ (called canonical vector field in \S \ref{study of geodesic equations}) on $\mathscr{H}$ is the unique null vector field on $\mathscr{H}$ that is tangent to $\mathscr{H}$,  hence also normal to $\mathscr{H}$.  The integral curves of $V$ in $\mathscr{H}$ are null pregeodesics.

\end{prop}


\subsection{Causal and vicious regions of the Kerr-star spacetime}

\begin{prop}[\cite{KBH_book},  Proposition $2.4.6$] \label{prop I and II casual} 
The BL blocks I and II are causal.
\end{prop}

\begin{cor}\label{causal region of K^*}
Let $K^*$ be the Kerr-star spacetime.  Then the region\\
 I$\;\cup$ II$\; \cup\;  \{r=r_{\pm}\} =\{ t^*\in\mathbb{R},  r\in[r_{-},+\infty),  (\theta,\phi^*)\in S^2\}\setminus \Sigma\subset K^*$ is causal.
\end{cor}

\begin{proof}
Let $\gamma$ be a future pointing curve.  If $\gamma$ is entirely contained either in I or in II,  then by Prop.  \ref{prop I and II casual},   $\gamma$ cannot be closed.   If $\gamma$ is entirely contained in $\mathscr{H}=\{r=r_{\pm}\}$ (closed totally geodesic null hypersurface of $K^*$ by Prop.  \ref{H is closed totally geod}),  then by Lem.  $1.5.11$ of \cite{KBH_book},  except for restphotons,  all other curves are spacelike,  but restphotons are integral curves of $V|_\mathscr{H}=(r^2_\pm+a^2)\partial_t+a\partial_\phi$,  which cannot be closed.  Since the time orientation $-\partial^*_r$ is null and transverse to the null hypersurface $\mathscr{H}$,  the future directed curves always go in the direction of $-\partial^*_r$,  if they hit $\mathscr{H}$ transversally.  Henceforth,  if $\gamma$ starts in the BL block I (II),  crosses $\mathscr{H}_{+}$ ($\mathscr{H}_{-}$) transversally,  enters the block II (III),  then  $\gamma$ cannot re-intersect $\mathscr{H}_{+}$  from II to I ($\mathscr{H}_{-}$ from III to II).  The last possibility is the following: $\gamma$ starts in I (II),  becomes tangent to $\mathscr{H}_{+}$ ($\mathscr{H}_{-}$),   hence either lies forever on $\mathscr{H}_{+}$ ($\mathscr{H}_{-}$) or leaves it at some point.  In the first case,  $\gamma$ is obviously not closed,  while in the second,  it cannot be closed because it will necessarily have to enter the region $\{r<r_{+}\}$ ($\{r<r_{-}\}$),  according to the time orientation. 
\end{proof}

\begin{prop}[\cite{KBH_book},  Proposition $2.4.7$]
The BL block III in the Kerr-star spacetime is vicious,  that is,  given any two points $p,q\in\,$III there exists a future directed timelike curve in III from $p$ to $q$.
\end{prop}

\begin{cor}
Let $p$ be a point in the BL block III of the Kerr-star spacetime.  Then there exists a closed timelike curve through $p$.
\end{cor}

\section{\textbf{Geodesics in Kerr spacetimes}}\label{study of geodesic equations}

\subsection{Constants of motion}

Let $(\mathcal{M}_{\textit{Kerr}},\mathbf{g})$ be a Kerr spacetime as in Def.  \ref{definition kerr spacetime}.  Recall that there are two Killing vector fields $\tilde{\partial}_t$ and $\tilde{\partial}_\phi$ on $\mathcal{M}_{\textit{Kerr}}$. 

\begin{defn}[\textit{Energy and angular momentum}] 
For a geodesic $\gamma$ of $(\mathcal{M}_{\textit{Kerr}},\mathbf{g})$,  the constants of motion 
\[
E=E(\gamma):=-\mathbf{g}(\gamma',\tilde{\partial}_t)
\]
and 
\[
L=L(\gamma):=\mathbf{g}(\gamma',\tilde{\partial}_\phi)
\]
are called its {\it energy} and its {\it angular momentum (around the axis of rotation of the black hole)}, respectively.
\end{defn}

\begin{defn}
For every BL block $\mathcal{B}_i$ define the \textit{canonical vector fields}
\[
V:=(r^2+a^2)\partial_t + a\partial_\phi\quad\text{ and }\quad W:=\partial_\phi + a \sin^2\theta\,\partial_t
\]
via the isometry $\Phi_i\colon  \mathcal{B}_i\hookrightarrow \mathcal{M}_{\textit{Kerr}}$.  
\end{defn}

\begin{oss}
$V$ and $W$ are not Killing vectors.
\end{oss}

\begin{defn}\label{definitions of P and D}
Let $\gamma$ be a geodesic in $\mathcal{M}_{\textit{Kerr}}$ with energy $E$ and angular momentum $L$.  Define the functions $\mathbb{P}$ and $\mathbb{D}$ along $\gamma$ by
\[
\mathbb{P}(r):=-\mathbf{g}(\gamma',V)=(r^2+a^2)E-La
\]
and 
\[
\mathbb{D}(\theta):=\mathbf{g}(\gamma',W)=L-Ea\sin^2\theta.
\]\\
\end{defn}

A geodesic in a Kerr spacetime has two additional constants of motions.  First,  there is the {\it Lorentian energy} $q:=\mathbf{g}(\gamma',\gamma')$,  which is always constant along every geodesic
in any pseudo-Riemannian manifold.  The second one is $K$,  which was first found by Carter in \cite{Carter_causality} using the separability of the Hamilton--Jacobi equation.  $K$ can be defined (see Ch.  $7$ in \cite{Chandrasekhar}) by

\[
K:=2\rho^2(r,\theta)\mathbf{g}(l,\gamma')\mathbf{g}(n,\gamma')+r^2q,
\]
where $l=\frac{1}{\Delta(r)}V+\partial_r$ and $n=\frac{1}{2\rho^2(r,\theta)}V-\frac{\Delta(r)}{2\rho^2(r,\theta)}\partial_r$.  See also \cite{Walker_Penrose} for a definition using a Killing tensor for the Kerr metric.  

\begin{defn}[\textit{Carter constant}]
On a Kerr spacetime,  the constant of motion
\[
Q:=K-(L-aE)^2\hspace{0.5cm}\textrm{or}\hspace{0.5cm} \mathcal{Q}:=Q/E^2\hspace{0.3cm} \textrm{if}\hspace{0.3cm}E\neq 0
\]
is called the Carter constant.  
\end{defn}

\subsection{Equations of motion}

\begin{prop}[\cite{KBH_book},  Proposition $4.1.5$,  Theorem $4.2.2$] \label{differential equations of geodessics}
Let $B$ be a BL block and $\gamma$ be a geodesic with initial position in $B\subset \mathcal{M}_{\textit{Kerr}}$ and constants of motion $E,  L,  Q,  q$.   Then the components of $\gamma$ in the BL coordinates $(t,r,\theta,\phi)$ satisfy the following set of \textit{first} order differential equations

\begin{align}
\begin{cases}
 \rho^2(r,\theta)\phi'=\frac{\mathbb{D}(\theta)}{\sin^2\theta}+a\frac{\mathbb{P}(r)}{\Delta(r)} \\ \rho^2(r,\theta) t'= a\mathbb{D}(\theta) + (r^2+a^2)\frac{\mathbb{P}(r)}{\Delta(r)} \label{geodes diff equations}\\ \rho^4(r,\theta) r'^2 = R(r)\\  \rho^4(r,\theta) \theta'^2 = \Theta (\theta)   
 \end{cases}
\end{align}
where

\begin{align*}
R(r):=&\Delta(r)\left[(qr^2-K(E,L,Q)\right]+\mathbb{P}^2(r)= \\ 
=& (E^2+q)r^4 -2Mqr^3 + \mathfrak{X}(E,L,Q) r^2 + 2MK(E,L,Q)r - a^2 Q,\label{other form of R(r)}\\
\Theta(\theta):=& K(E,L,Q)+qa^2 \cos^2\theta -\frac{ \mathbb{D}(\theta)^2}{\sin^2\theta}= \\
=& Q + \cos^2\theta \left[ a^2(E^2+q)-L^2/\sin^2\theta\right],
\end{align*}
with
\[
\mathfrak{X}(E,L,Q):=a^2(E^2+q)-L^2-Q\text{, and } K(E,L,Q)=Q+(L-aE)^2.
\]

\end{prop}

\begin{oss}\label{non negativitivity of polynomials}
Since in the third and in the fourth differential equations of Prop.  \ref{differential equations of geodessics} the left-hand sides are clearly non-negative,  we see that the polynomials $R(r)$ and $\Theta(\theta)$ are non-negative along the geodesics.  Hence the geodesic motion can only happen in the $r,\theta$-region for which $R(r),\Theta(\theta)\geq 0$.  
\end{oss}

In order to study geodesics that cross the horizons 
\[
\mathscr{H}=\{\Delta(r)=0\}=\{r=r_{\pm}\},  
\]
it is necessary to introduce the Kerr-star coordinate system.  Note however that since the change of coordinates modifies only the $t$ and the $\phi$ coordinates and the $r,\theta$-differential equations do not involve $t$ and $\phi$,  the last two differential equations do extend over $\mathscr{H}$.  Observe also that the $r,\theta$- differential equations are not singular on $\mathscr{H}$,  while the $t,\phi$-differential equations are. \\

Notice that $\Theta(\theta)$ is also well-defined if the geodesic crosses $A=\{\theta=0,\pi\}$.  Indeed,  $L=0$ (because  $\tilde{\partial_\phi}\equiv 0$ on $A$),  hence $\mathbb{D}(\theta)=-Ea\sin^2\theta$,  and then
  
\begin{align*}
\Theta(\theta)&=K(E,0,Q) +qa^2\cos^2\theta-(-Ea\sin^2\theta)^2/\sin^2\theta=Q+a^2E^2+qa^2\cos^2\theta-a^2E^2\sin^2\theta\\
&=Q+a^2(E^2+q)\cos^2\theta.\\
\end{align*}
Thus the $r,\theta$-differential equations can be used to study geodesics on the whole Kerr-star spacetime.

\begin{oss}
The system \eqref{geodes diff equations} is composed of first order differential equations,  while the geodesic equation is \textit{second} order.  There exist solutions of \eqref{geodes diff equations},  called \textit{singular},  which do not correspond to geodesics.  For example,  if $r_0\in\mathbb{R}$ is a multiplicity one zero of $r\mapsto R(r)$,  then $r_0$ solves the radial equation in \eqref{geodes diff equations},  since in this case $r'(s)=0$ for all $s$,  but we do not have a geodesic. 
\end{oss}

\subsection{Dynamics of geodesics}

The non-negativity of $R(r)$ and $\Theta(\theta)$ in the first order differential equations of motion \eqref{geodes diff equations} can be used to study the dynamics of the $r,\theta$-coordinates of the geodesics,  together with the next proposition.

\begin{prop}[\cite{KBH_book},  Corollary $4.3.8$]\label{initial conditions and zeroes}
Suppose $R(r_0)=0$.  Let $\gamma$ be a geodesic whose $r$-coordinate satisfies the initial conditions $r(s_0)=r_0$ and $r'(s_0)=0$.
\begin{enumerate}
\item If $r_0$ is a multiplicity one zero of $R(r)$,  \textit{i.e.} $R'(r_0)\neq 0$,  then $r_0$ is an $r$-turning point,  namely $r'(s)$ changes sign at $s_0$.
\item If $r_0$ is a higher order zero of $R(r)$,  \textit{i.e.} at least $R'(r_0)= 0$,  then $\gamma$ has constant $r(s)=r_0$.  
\end{enumerate}
Analogous results hold for $r$ and $R(r)$ replaced by $\theta$ and $\Theta(\theta)$.  
\end{prop}

\section{\textbf{Properties of timelike geodesics in Kerr spacetimes}}\label{section: properites of timelike geodesics}

\subsection{Principal geodesics}

Since the vector fields $V, W, \partial_r,\partial_\theta$ are linearly independent,  the tangent vector to a geodesic $\gamma$ can be decomposed as $\gamma'=\gamma'_\Pi + \gamma'_\perp$ where $\Pi:= span \{ \partial_r,  V \}$ (timelike plane) and $\Pi^\perp:= span \{ \partial_\theta,  W \}$ (spacelike plane).


 


\begin{defn}
A Kerr geodesic $\gamma$ is said to be \textit{principal} if $\gamma' = \gamma'_\Pi$. 
\end{defn}

\begin{prop}[\cite{KBH_book},  Corollary $4.2.8(1)$] \label{eq principal nulls}
If $\gamma$ is a timelike geodesic,  then $K\geq 0$,  and\\ $K=0\Longleftrightarrow\;$ $\gamma$ is a principal geodesic in the $Eq=\{\theta=\pi/2\}$.
\end{prop}




\subsection{Timelike geodesics with $Q<0$}

\begin{prop}\label{Q<0 condition}
Let $\gamma$ be a timelike ($q<0$) geodesic with $Q<0$.  Then

\begin{enumerate}
\item $\gamma$ does not intersect $Eq=\{\theta=\pi/2\}$;
\item $a^2(E^2+q)> L^2$ and in particular $E\neq 0$ and $E^2+q>0$.
\end{enumerate}
\end{prop}

\begin{proof}
If $\gamma\cap A=\emptyset$,  then from the $\theta$-equation of \eqref{geodes diff equations} we have

\begin{align*}
\cos^2\theta [L^2/\sin^2\theta-a^2(E^2+q)]=Q-\rho^4(r,\theta)\theta'^2<0.
\end{align*}
Hence $\cos^2\theta\neq 0$ and $L^2/\sin^2\theta-a^2(E^2+q)<0$,  hence $\gamma\cap Eq=\emptyset$ and $a^2(E^2+q)>L^2$,  so $E\neq 0$ and $E^2+q>0$.

If $\gamma\cap A\neq\emptyset$,  then  $L=0$ since $\tilde{\partial}_\phi \equiv 0$ on $A$ and 

\begin{align*}
-a^2(E^2+q)\cos^2\theta =Q-\rho^4(r,\theta)\theta'^2<0.
\end{align*}
Therefore $\cos^2\theta\neq 0$ and $a^2(E^2+q)>0=L^2$,  $E\neq 0$ so $\gamma\cap Eq=\emptyset$.
\end{proof}


\begin{prop}[\cite{KBH_book},  Corollaries $4.9.2,4.9.3$]\label{geod Q<0}
For $Q<0$ timelike geodesics,  $R(r)$ is convex and has either zero or two negative roots,  which may be coincident.  Therefore the only possible bounded $r$-behaviour is $r(s)=\textrm{const}<0$.
\end{prop}

\section{\textbf{Proof of Theorem  \ref{main theorem}}} \label{section: main theorem}

\subsection{Strategy of the proof} \label{strategy of the proof}
The following argument is similar to the one used to prove the nonexistence of closed null geodesics in \cite{sanzeni2024non}.  Let $\gamma\colon I \to K^*$ be a closed timelike geodesic (CTG).
Since the radius function $r\colon K^*\to\mathbb{R}$ is everywhere smooth the composition $r\circ \gamma$ has at least two critical points $s_0<s_1$ in each period $[a,a+T)$, i.e.  
$(r\circ \gamma)'(s_0)=(r\circ \gamma)'(s_1)=0$. Since $\rho\colon K^*\to\mathbb{R}$  does not vanish on $K^*$ the differential equation for $r\circ \gamma$
\[
(\rho\circ \gamma)^4[(r\circ \gamma)']^2=R(r\circ \gamma)
\]
implies that $R(r\circ \gamma(s_{0,1}))=0$.  Because of the differential equation,  the geodesic motion must happen in the $r$-region on which $R(r\circ \gamma)\geq 0$.  Further since $R$ is a polynomial in $r$ we can distinguish two cases:
\begin{enumerate}

\item The zeros $r\circ \gamma(s_{0,1})$ of $R$ are simple, i.e. $dR/dr\neq 0$ at these points. Then $r\circ \gamma(s_{0,1})$ are turning points of $r\circ \gamma$, i.e. 
$(r\circ \gamma)'$ changes its sign at $s_0$ and $s_1$.  

\item One of the zeros $r\circ \gamma(s_{0})$ or $r\circ \gamma(s_{1})$ is a higher order zero of $R$. Then $r\circ \gamma$ is constant.
\end{enumerate}
Both the two facts follow from Proposition \ref{initial conditions and zeroes}.  

Most possible CTGs can be ruled out by comparing the location of the zeros of $R(r)$ with the following consequence of the causal structure of Kerr:

\begin{lem}\label{lemma about closed curves confined}
Let $\gamma\colon  I\to K^*$ be a closed timelike geodesic. Then $r\circ \gamma\subset \{r<r_-\}$.
\end{lem}

\begin{proof}
The region 
\[
\{r\geq r_{-}\}=\{ t^*\in\mathbb{R},  r\in[r_{-},+\infty),  (\theta,\phi^*)\in S^2\} \setminus \Sigma\subset K^*
\] 
is causal by Corollary  \ref{causal region of K^*} and closed timelike geodesics cannot intersect $\{r=r_{-}\}$ by Prop.  \ref{oss event horizon}. 
\end{proof}

\begin{prop}\label{prop spacelike foliation}
There are no closed timelike geodesics strictly contained in $\{0<r<r_{-}\}$.
\end{prop}

\begin{proof}
First we claim that the hypersurfaces $\mathcal{N}_t:=\{t=\textrm{const}\}\cap\{0<r<r_{-}\}$ are spacelike.  Indeed,  if $p\in\mathcal{N}_t\setminus A$,  where $A=\{\theta=0,\pi\}$,  then $T_p\mathcal{N}_t$ is spanned by $\partial_r,\partial_\theta,\partial_\phi$ which are spacelike and orthogonal to each other.  If $p\in A\subset\mathcal{N}_t$,  then $p=(t,r,q)$ with $q=(0,0,\pm 1)\in S^2\subset \mathbb{R}^3$,  and we may replace $\partial_\theta,\partial_\phi$ by any basis of $T_qS^2$.  Suppose by contradiction that there exist a CTG $\gamma$ in $\{0<r<r_{-}\}$.  Then $t\circ \gamma$ takes values in a closed $t$-interval,  and so must attain maximum $t_0$,   say at parameter $s_0$.  Therefore  $\gamma'(s_0)\in T_{\gamma(s_0)}\mathcal{N}_{t_0}$,  so $\gamma$ must be tangent $\mathcal{N}_{t_0}$.  This is a contradiction since $\gamma'(s_0)$ is timelike.
\end{proof}

\hspace{1cm}





\subsection{Horizons and Axis cases}

First we rule out CTGs entirely contained in the axis $A=\{\theta=0,\pi\}$ and CTGs intersecting  the horizon $\mathscr{H}=\{r=r_{\pm}\}$.

\subsubsection*{The case of the horizon $\mathscr{H}=\{r=r_{\pm}\}$}

\begin{prop}\label{oss event horizon}
There are no CTGs intersecting $\mathscr{H}=\{r=r_{\pm}\}$.
\end{prop}

\begin{proof}

First,  notice that there are no timelike geodesics entirely contained in $\mathscr{H}$ by Prop.  \ref{H is closed totally geod}.  Consider now a timelike geodesic intersecting $\mathscr{H}$ transversally.  Since each connected component $\{r=r_\pm\}$ of $\mathscr{H}$ is an orientable hypersurface separating the orientable manifold $K^*$,  every closed curve transversal to $\mathscr{H}$ has to intersect $\{r=r_\pm\}$ an even number of times.  Further since $K^*$ is time-oriented by $-\partial_r^*$,  all tangent vectors to a timelike geodesic transversal to $\mathscr{H}$ have to lie on one side of $\mathscr{H}$.  Therefore a timelike geodesic transversal to $\mathscr{H}$ can intersect each connected component $\{r=r_\pm\}$ only once.  This shows that no timelike geodesic transversal to $\mathscr{H}$ can close.
\end{proof}

\subsubsection*{The case of the axis $A=\{\theta=0,\pi\}$}

\begin{prop} \label{no closed geodesics in axis}
There are no CTGs which are tangent at some point to $A=\{\theta=0,\pi\}$.  In particular,  there are no CTGs entirely contained in $A$.
\end{prop}

\begin{proof}
First of all,  $A=\{\theta=0,\pi\}$ is a $2$-dimensional closed totally geodesic submanifold by Lem.  \ref{A and Eq closed totally geod subman}.    Hence if a geodesic $\gamma$ is tangent to $A$ at some point,   it will always lie on $A$.  If $\gamma\in A$,  then $L=0$,  since $\tilde{\partial}_\phi \equiv 0$ on $A$.  From the $\theta$-equation of Prop. \ref{differential equations of geodessics},  we have $Q=-a^2(E^2+q)$,  hence $K=-a^2q$.  Therefore we obtain

\begin{align*}
R(r)=\Delta(r)(qr^2+a^2q)+(r^2+a^2)^2E^2=(r^2+a^2)\big[ (r^2+a^2)E^2 + q\Delta(r) \big].
\end{align*}
Now distinguish the cases $E=0$ and $E\neq 0$.  If $E=0$,  $R(r)=q(r^2+a^2)\Delta(r)$,  so the only turning points must be on $\mathscr{H}$.  Hence this polynomial cannot produce a CTG since the hypersurfaces $\mathscr{H}=\{ r=r_{\pm}\}$ are closed totally geodesic submanifolds by Prop. \ref{H is closed totally geod} and a geodesic cannot have turning points on such hypersurfaces because it would be tangent to them there.  If $E\neq 0$,  

\begin{align*}
R(r)=(r^2+a^2) \big[  r^2(E^2+q) -2Mqr + a^2(E^2+q)  \big].
\end{align*}
Let us consider the discriminant of the second factor:  $\textrm{dis}=4M^2q^2-4a^2(E^2+q)^2$.  If $\textrm{dis}<0$,  no bounded $r$-behaviour is possible.  If $\textrm{dis}=0$,  the two coincident roots of $R(r)$ are $Mq/(E^2+q)=:\bar{r}$.  Since $\mathscr{H}$ is a null hypersurface by Prop. \ref{H is closed totally geod},  there are no timelike geodesics in $\mathscr{H}$,  hence we may assume $\bar{r}\neq r_{\pm}$.  Therefore we can use the $t$-diff.  equation of Prop.  \ref{differential equations of geodessics} to get

\begin{align*}
\rho^2(r,\theta)t'=(r^2+a^2)^2E/\Delta(r)\neq 0,
\end{align*}
for every $r$.  Hence $t(s) $ must be monotone and the geodesic cannot be closed. 

If $\textrm{dis}>0$,  a bounded $r$-behaviour would require $E^2+q<0$.  However the following root satisfies

\begin{align*}
\frac{Mq}{E^2+q}+\sqrt{\frac{M^2q^2}{(E^2+q)^2}-a^2}>r_{+}=M+\sqrt{M^2-a^2}, 
\end{align*}
since $0<\frac{q}{E^2+q}=\big( 1+ \frac{E^2}{q}\big)^{-1}>1$,  which contradicts Lemma \ref{lemma about closed curves confined}.

\end{proof}

\subsection{Steps of the proof for other cases}\label{steps of other cases}

The proof splits into two main cases  $E=0$ and $E\neq 0$.\\  

If $E=0$ (\S \ref{section $E=0$}),  by Prop.  \ref{eq principal nulls}  we can analyse the only two possible subcases $K(0,L,Q)=0$  (\S \ref{E=0,K=0}) and $K(0,L,Q)>0$ (\S \ref{E=0,K>0}).\\

If $E\neq 0$ (\S \ref{section $E div 0$}),  we analyse three subcases $Q=0$ (\S \ref{case Q=0}),  $Q>0$ (\S \ref{case Q>0}) and $Q<0$ (\S \ref{case Q<0}).\\

\begin{oss}
The only case which requires a detailed analysis of the differential equations is the case $E\neq 0$ and $Q<0$ (see \ref{case Q<0}).
\end{oss}

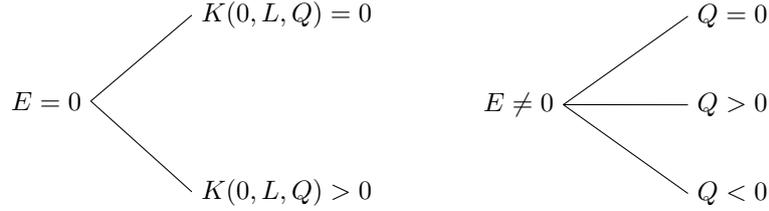
\begin{figure}[h]
\centering
\begin{tikzpicture}[grow=right]
\tikzset{level distance=90pt,sibling distance=50pt}
 \tikzset{frontier/.style={distance from root=170pt}}

\Tree 
[.$E=0$ [.$K(0,L,Q)>0$  ]
                       [.$K(0,L,Q)=0$ ] 
                       ]

\end{tikzpicture}
\hspace{1cm}
\begin{tikzpicture}[grow=right]
\tikzset{level distance=80pt,sibling distance=18pt}
 \tikzset{frontier/.style={distance from root=150pt}}

\Tree 
               [.$E\neq 0$  [.$Q<0$ ] 
               		              [.$Q>0$ ] 
                                    [.$Q=0$ ] 
                ]

\end{tikzpicture}
\vspace*{5mm}
\caption{All the geodesic types which have to be studied.}
\label{figure steps of proof}
\vspace*{-5mm}
\end{figure}

\subsection{Case $E=0$}\label{section $E=0$}

\vspace{1cm}

From Prop.  \ref{differential equations of geodessics},  we have

\begin{align}
\label{R equation E=0}R(r)&=qr^4-2Mqr^3+\mathfrak{X}(0,L,Q)r^2+2MK(0,L,Q)r-a^2Q\geq 0,\\
\label{theta equation E=0}\Theta(\theta)&=Q+\cos^2\theta\bigg( a^2q-\frac{L^2}{\sin^2\theta} \bigg)\geq 0,
\end{align}
with $\mathfrak{X}(0,L,Q)=a^2q-L^2-Q$ and $K(0,L,Q)=L^2+Q$.  Notice that we must have $Q\geq 0$ by \eqref{theta equation E=0},  hence $\mathfrak{X}(0,L,Q)<0$.  
\vspace{1cm}

\subsubsection{Subcase $K(0,L,Q)=0$} \label{E=0,K=0} Then $Q=-L^2\leq 0$.  Since $Q\geq 0$,  we hence must have $Q=L=0$. Therefore

\begin{align}
R(r)=qr^2\Delta(r).
\end{align}
Then by Lemma \ref{lemma about closed curves confined},  the only possible $r$-behaviour for a CTG would be $r(s)=\textrm{const}=0$.  However,  this geodesic would lie on the ring singularity by Prop.  \ref{eq principal nulls}.


\vspace{0.8cm}

\subsubsection{Subcase $K(0,L,Q)>0$}  \label{E=0,K>0}

The signs of the coefficients of $R(r)$ are $-\; +\; -\; +\; -$ if $Q>0$ (respectively $-\; +\; -\; +$ if $Q=0$),  hence there are no roots in $r<0$ and either four or two or zero positive roots (respectively either three or one positive roots) by the \textit{"Descartes' rule of signs"}.  Therefore a CTG $\gamma$ could have a bounded $r$-behaviour only if $r(s)\in [0,r_{-})$ by Lemma \ref{lemma about closed curves confined}.  Now distinguish the cases $Q=0$ and $Q>0$.  If $Q=0$,  then $\cos^2\theta(s)=0$ by \eqref{theta equation E=0},  hence $\gamma$ lies in $Eq$.  Therefore $r(s)>0$ because otherwise $\gamma$ would hit the ring singularity.  If $Q>0$,  then $R(0)<0$,  hence $r(s)>0$.  By Lemma \ref{lemma about closed curves confined},  in both cases we have $0<r(s)<r_{-}$.  So $\gamma$ cannot be closed by Prop.  \ref{prop spacelike foliation}.



\subsection{Case $E\neq 0$}\label{section $E div 0$}

\vspace{0.8cm}
\subsubsection{Subcase $Q=0$} \label{case Q=0}

We have

\begin{align}
R(r)&=(E^2+q)r^4-2Mqr^3+\mathfrak{X}(E,L,0)r^2+2MK(E,L,0)r\geq 0,\\
\label{theta-eq Ediv 0,Q=0} \Theta(\theta)&=\cos^2\theta\bigg[ a^2(E^2+q)-\frac{L^2}{\sin^2\theta} \bigg]\geq 0. 
\end{align}

\begin{prop}\label{prop bounded timelike in Eq}
All the timelike geodesics with $E\neq 0, \, Q=0$ which have bounded $r$-behaviour lie in $Eq=\{\theta=\pi/2\}$.
\end{prop}

\begin{proof}
Suppose there exists an $r$-bounded timelike geodesic with $E\neq 0,  Q = 0$ for which $\theta(s) \neq \pi/2$ for some $s$.  Then from \eqref{theta-eq Ediv 0,Q=0},  we get

\begin{align*}
a^2(E^2+q)\geq \frac{L^2}{\sin^2\theta}\geq L^2.
\end{align*}
So $\mathfrak{X}(E,L,0)=a^2(E^2+q)-L^2\geq 0$,  hence $E^2+q\geq 0$.  Observe also that $K(E,L,0)\geq 0$ by Prop.  \ref{eq principal nulls}.  Since the coefficients of $R(r)$ are all non-negative,  $R(r)$ cannot have positive roots by the \textit{"Descartes' rule of signs"}.  By Prop.  $4.8.2$ of \cite{KBH_book} there are no timelike geodesics with bounded $r$-behaviour in $r<0$.  Therefore two possibilities are left,  either  $r(s)\in [\bar{r},0]$,  with $\bar{r}<0$ or $r(s)=\textrm{const}=0$.
In the first case we would have $R'(0)<0$,  which contradicts $K(E,L,0)\geq 0$.  In the second case we would have $R'(0)=0$,  hence $L=aE$,  which contradicts $\mathfrak{X}(E,L,0)\geq 0$.
\end{proof}
By Prop.  \ref{prop bounded timelike in Eq},  we can suppose $\theta(s)= \pi/2$ for every $s$. Since the geodesics are constrained in $Eq$,  $r = 0$ cannot be reached,  hence the $r$-motion must be either confined in $\{r<0\}$ or in $\{r>0\}$.  By Prop.  $4.8.2$ of \cite{KBH_book} no bounded $r$-behaviour in $r<0$ is allowed.  By Lemma \ref{lemma about closed curves confined} the $r$-motion must then be constrained in the region $\{0<r<r_{-}\}$.  Such geodesic cannot be closed by Prop.  \ref{prop spacelike foliation}.




\vspace{1cm}

\subsubsection{Subcase $Q>0$} \label{case Q>0} We have $R(0)=-a^2Q<0$.  Hence a bounded $r$-behaviour is either confined in $\{r<0\}$ or in $\{r>0\}$.  However if the timelike geodesic lies in $\{r>0\}$,  it must be constrained in $\{0<r<r_{-}\}$ by Lemma \ref{lemma about closed curves confined} and it cannot be closed by Prop.  \ref{prop spacelike foliation} .  If it lies in $\{r<0\}$,  by Prop.  $4.8.2$ of \cite{KBH_book}  it must be a fly-by geodesic, \textit{i.e.} $r\circ\gamma\subset [-\infty,r_{\textit{turn}}]$,  where at $r_{\textit{turn}}$ the geodesic reverse its $r$-motion.

\vspace{1cm}

\subsubsection{Subcase $Q<0$} \label{case Q<0}

This is the last remaining case and the most difficult one.  By Prop.  \ref{geod Q<0},  the only possible bounded behaviour is $r(s)=\textrm{const}<0$ (see  Fig.\ref{R(r) Q<0 r=const}).  Such geodesics  are known in the literature as \textit{spherical geodesics},  see e.g.  \cite{spherical_photon}.

\begin{figure}[H]
\centering
\includegraphics[scale=0.6]{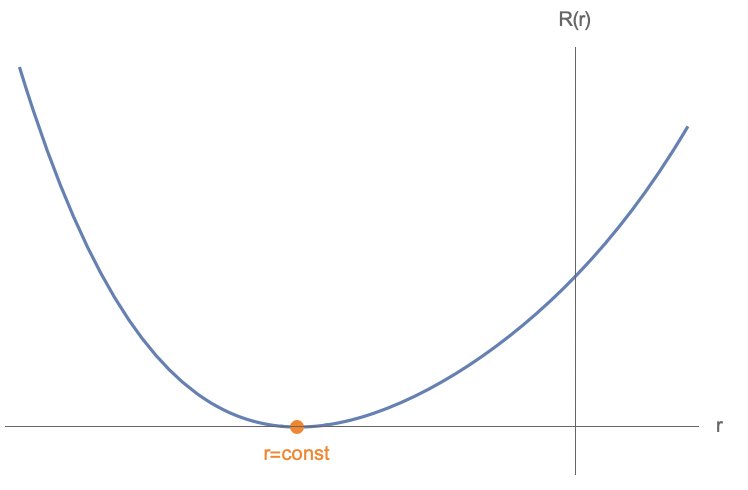} 
\caption{Plot of $R(r)$ with $a=3,\; M=5,\; q\approx -0.337\;E\approx 1.156, \; L\approx 0.130,\;  Q=-6$.}
\label{R(r) Q<0 r=const}
\vspace*{-5mm}
\end{figure}
\hspace{0.5cm}

\begin{prop}\label{spherical timelike geodesics have r>-M}
Consider a spherical timelike geodesic $\gamma$ with $Q<0$ at radius $r$.  If $r\leq -M$,   then $\gamma$ cannot be closed.
\end{prop}

\begin{proof}
First observe that if $r\leq -M$ and $\theta\neq 0,\pi$,  then $\partial_\phi$ is spacelike.  Indeed the sign of $\mathbf{g}(\partial_\phi,\partial_\phi)$ is determined by the sign of the following function


\begin{align*}
r^4+a^2(1+\cos^2\theta)r^2+2a^2M\sin^2\theta\; r + a^4\cos^4\theta  & \geq  r^4+a^2r^2 + 2a^2Mr\\
& \geq M^2 r^2  +2a^2 M r +a^2 M^2\\
& >0,
\end{align*}
where we used respectively the bounds on the trigonometric functions,  $r\leq -M<0$,   and the negativity of the discriminant of the $r$-polynomial since $|a|<M$.\\

Now we claim that the hypersurfaces $\mathcal{S}_t:=\{t=\textrm{const}\}\cap \{r\leq -M \}$ are spacelike,  hence $t\circ\gamma$ is monotonic,  therefore $\gamma$ cannot be closed.  Indeed,  if $p\in\mathcal{S}_t\setminus A$,  where $A=\{\theta=0,\pi\}$,  then $T_p\mathcal{S}_t$ is spanned by $\partial_r,\partial_\theta,\partial_\phi$ which are spacelike and orthogonal to each other.  If $p\in A\subset\mathcal{S}_t$,  then $p=(t,r,q)$ with $q=(0,0,\pm 1)\in S^2\subset \mathbb{R}^3$,  and we may replace $\partial_\theta,\partial_\phi$ by any basis of $T_qS^2$. 
\end{proof}

\hspace{1cm}

By Prop.  \ref{Q<0 condition},  timelike geodesics with negative Carter constant do not meet $Eq=\{\theta=\pi/2\}$,  hence $\cos^2\theta\neq 0$.  Then we may define $u:=\cos^2\theta\in (0,1]$.  Since $E^2+q> 0$ by Prop.  \ref{Q<0 condition},  we can re-write the $\theta$-equation in \eqref{geodes diff equations} as

\begin{align}\label{theta di u}
\bigg( \frac{\rho^2(r,u)}{\sqrt{E^2+q}}\bigg)^2\frac{(u')^2}{4u}=-a^2u^2+(a^2-\hat{\Phi}^2-\hat{\mathcal{Q}})u+\hat{\mathcal{Q}}=:\hat{\Theta}(u),
\end{align}
where $\hat{\Phi}:=L/\sqrt{E^2+q}$ and $\hat{\mathcal{Q}}:=Q/(E^2+q)$.  Since we must have $\hat{\Theta}(u)\geq 0$ somewhere in $(0,1]$,  $w:=a^2-\hat{\Phi}^2-\hat{\mathcal{Q}}>0$ because $\hat{\mathcal{Q}}<0$ and the coefficient of the second order term is negative.  Therefore $\hat{\Theta}$ must have roots given by

\begin{align}\label{formula for u pm}
u_{\pm}=\frac{w\pm\sqrt{\textrm{dis}}}{2a^2}
\end{align}
where $\textrm{dis}:=w^2+4a^2\hat{\mathcal{Q}}$,  so that 

\begin{align}
\hat{\Theta}(u)=-a^2(u-u_{+})(u-u_{-}).
\end{align}
Hence we have

\[
0<u_{-}\leq u_{+}.
\]



Note that we must have $u_{-}\leq 1$,  otherwise the downward parabola function $\hat{\Theta}(u)$ cannot be non-negative somewhere in $(0,1]$.  Since $\hat{\Theta}(u)$ is quadratic and $\hat{\Theta}(1)=-\hat{\Phi}^2\leq 0$,  we must either have $u_{+}\leq 1$ or $u_{-}=1$.  However in the latter case  $\theta\circ\gamma(s)= 0, \pi$ for all $s$,  so the geodesic lies in $A$ and it cannot be closed by Prop.  \ref{no closed geodesics in axis}.  By hypothesis,  we can hence assume

\begin{align*}
0<u_{-}\leq u_{+}\leq 1.
\end{align*}

\begin{prop}\label{u-=u+ Q<0 prop}
In the Kerr-star spacetime,  consider a timelike geodesic $\gamma$ with $\hat{\mathcal{Q}}<0$ and $r=\textrm{const}$.  If $\textrm{dis}=0$,  then $\theta=\textrm{const}$ and the geodesic cannot be closed.
\end{prop}

\begin{proof}
Since $\textrm{dis}=0$,  we have $u_{-}=u_{+}$.  Then

\[
\hat{\Theta}(u)=-a^2\big(u-u_{+}\big)^2\geq 0.
\]
Hence the last inequality is satisfied only if $u=u_{+}=\textrm{const}$,  therefore $\theta=\textrm{const}$.  Then the geodesic $\gamma$ cannot be closed.  Indeed,  there are two possibilities.  First,  if $\gamma$ is entirely contained in $A$,  then it cannot be closed by Prop.  \ref{no closed geodesics in axis}.  Second,  if $\gamma$ is not entirely contained in $A$,  by Prop.  \ref{differential equations of geodessics}  a geodesic of the form $s\mapsto (t(s),r_0,\theta_0,\phi(s))$ has $t'\equiv\textrm{const}$ and $\phi'\equiv\textrm{const}$.  It follows that $s\mapsto t(s)$ and $s\mapsto \phi(s)$ are affine functions.  If the geodesic is bounded in $K^*$,  then $t(s)$ must be constant.  For $\gamma(s)=(t_0,r_0,\theta_0,b_0s+b_1)$,  $b_0,b_1\in\mathbb{R}$,   the geodesic equation can be written in BL coordinates as $\Gamma^\alpha_{\phi\phi}(\gamma(s))b_0^2=0$.  We claim that the Christoffel symbol $\Gamma^{\theta}_{\phi\phi}$ cannot vanish on $\gamma$.  Indeed,

\begin{align*}
\Gamma^{\theta}_{\phi\phi}&=-\frac{\sin\theta \cos\theta}{\rho^6(r,\theta)}\bigg[ \rho^4(r,\theta) \frac{\mathbf{g}(\partial_\phi,\partial_\phi)}{\sin^2\theta} +2M(r^2+a^2)a^2r\sin^2\theta \bigg]\neq 0
\end{align*}
since $\theta\neq 0,\pi$ because we have already ruled out closed timelike geodesics in $A$,  $\theta\neq \pi/2$ by Prop.  \ref{Q<0 condition},  $\dot{\gamma}=\partial_\phi$ is timelike,  and $r<0$.  Hence,  $b_0=0$ and the geodesic degenerates to a point.
\end{proof}

\begin{oss}
Closed timelike curves exist in the Kerr-star spacetime: for instance,  they are given by the integral curves of the vector field $\partial_\phi$,  whenever this last happens to be timelike for some negative $r$.  Such curves cannot be geodesics by Prop.  \ref{u-=u+ Q<0 prop}.
\end{oss}

We may now assume $\textrm{dis}>0$.  Therefore we have the following chain of inequalities

\begin{align}
0<u_{-}<u_{+}\leq 1.
\end{align}

We hence define 

\begin{align}\label{values of the thetai}
\theta_1:=\arccos(\sqrt{u_{+}}),\theta_2:=\arccos(\sqrt{u_{-}}),\theta_3:=\arccos(-\sqrt{u_{-}}),\theta_4:=\arccos(-\sqrt{u_{+}})
\end{align}
so that
\begin{align}
0\leq \theta_1< \theta_2<\frac{\pi}{2}<\theta_3< \theta_4\leq \pi.
\end{align}

\begin{prop}\label{prop theta behaviours Q<0}
In the Kerr-star spacetime,  timelike geodesics with $\hat{\mathcal{Q}}<0$,  $r=\textrm{const}$  and $\theta\neq\textrm{const}$ can have one of the following $\theta$-behaviours:

\begin{itemize}
\item if $0<u_{-}<u_{+}<1$,  then the $\theta$-coordinate oscillates periodically in one of the following intervals $0<\theta_1\leq \theta\leq \theta_2<\pi/2$ or $\pi/2<\theta_3\leq \theta\leq \theta_4<\pi$;
\item if $0<u_{-}<u_{+}=1$,  then $\hat{\Phi}(=L/\sqrt{E^2+q})= 0$ and the $\theta$-coordinate oscillates periodically in one of the following intervals $0=\theta_1\leq \theta\leq \theta_2<\pi/2$ or $\pi/2<\theta_3\leq \theta\leq \theta_4=\pi$,
\end{itemize}
where the $\theta_i$,  $i=1,2,3,4$,  are given by \eqref{values of the thetai}.
\end{prop}

\begin{proof}
See the proof of Prop.  $5.9$ of \cite{sanzeni2024non}. 
\end{proof}

Consider the first order equations of motion (with the rescaled constants of motion\\
 $\hat{\mathcal{Q}}:=Q/(E^2+q),  \hat{\Phi}:=L/\sqrt{E^2+q},  \Phi:=L/E$),  for a constant $r<0$:

\begin{align}
\label{equation of theta plot}\frac{\rho^2(r,\theta)}{\sqrt{E^2+q}}\frac{d\theta}{ds}&=\pm\sqrt{\Theta(\theta)}=\pm\sqrt{\hat{\mathcal{Q}}+a^2\cos^2\theta-\hat{\Phi}^2\frac{\cos^2\theta}{\sin^2\theta}}\\
\label{equation of t plot}\frac{\rho^2(r,\theta)}{E}\frac{dt}{ds}&=\frac{r^2+a^2}{\Delta(r)}(r^2+a^2-a\Phi)+a(\Phi-a\sin^2\theta),
\end{align}
where now the function $\Theta(\theta)$ is meant as the ratio of the $\Theta$-function appearing in  Prop.  \ref{differential equations of geodessics} and $(E^2+q)$.  Because of the $\theta$-differential equation,  we can restrict to an interval $\mathcal{U}\subset\theta^{-1}\big( (\theta_1,\theta_2)\big)$ on which $d\theta/ds$ is either everywhere positive or everywhere negative (depending on the initial condition).  Due to the symmetry in \eqref{equation of theta plot} and the fact that $r=\textrm{const}$,  $\theta(s)$ is periodic over twice the interval $\mathcal{U}$.  For instance,  set $\mathcal{U}=(0,T/2)$,   starting from $\theta(0)=\theta_1$,  hence $\theta'(s)=+\sqrt{\Theta(\theta)} >0$ for $s\in(0,T/2)$,  then $\theta'(s)=-\sqrt{\Theta(\theta)} <0$ for $s\in(T/2,T)$,  where $\theta'(T/2)=0$,  because $\theta(T/2)=\theta_2$ ($'\equiv d/ds$) and Prop.  \ref{initial conditions and zeroes},  which explains the change of sign of $\theta'(s)$ (using the fact that $\theta_1,\theta_2$ are multiplicity one zeroes of $\Theta(\theta)$).  Hence every $\Delta s=T/2$,  $\theta'(s)$ changes sign.  See Figures $11,12,13,14$ of \cite{sanzeni2024non} for analogous $\theta$-motions.

At parameters where $\Theta(\theta)\neq 0$ we can combine \eqref{equation of t plot} and \eqref{equation of theta plot} to get

\begin{align}\label{t theta eq}
\frac{\sqrt{E^2+q}}{E}\frac{dt}{d\theta}=\frac{r^2\Delta(r)+2Mr(r^2+a^2-a\Phi)}{\pm\Delta(r)\sqrt{\Theta(\theta)}}+a^2\frac{\cos^2\theta}{\pm\sqrt{\Theta(\theta)}}=B(r,a,\Phi)\frac{1}{\pm\sqrt{\Theta(\theta)}}+a^2\frac{\cos^2\theta}{\pm\sqrt{\Theta(\theta)}},
\end{align}
with $B(r,a,\Phi):=\frac{r^2\Delta(r)+2Mr(r^2+a^2-a\Phi)}{\Delta(r)}.$\\

\begin{lem}
Consider a spherical timelike geodesic $\gamma$ at radius $r$ with negative Carter constant $Q$ and angular momentum $\Phi$.  If $B(r,a,\Phi)\geq 0$,  then $\gamma$ cannot be closed.
\end{lem}

\begin{proof}
The equation \eqref{equation of t plot} can be written as 
\[
\frac{\rho^2(r,\theta)}{E}\frac{dt}{ds}= B(r,a,\Phi)+a^2\cos^2\theta.
\]
Then the $t$-coordinate is monotonic,  hence it is non-periodic. 
\end{proof}

\begin{prop}\label{prop spherical geodesics}
In a Kerr spacetime,  a timelike geodesic with negative Carter constant $\mathcal{Q}$ and constant radial coordinate has one of the two following pairs $\big( \Phi_{\pm}, \mathcal{Q}_{\pm}\big)$ of constants of motion given by

\begin{align}
\Phi_{\pm}&=\Phi_{\pm}(r,q)=\frac{E^2M(a^2-r^2)\pm \sqrt{f(r,q)}}{a E^2(M-r)},  \nonumber \\ 
\textrm{} \label{class r=const}\\
\nonumber \mathcal{Q}_{\pm}&=\mathcal{Q}_{\pm}(r,q)= \frac{r^2}{a^2E^2(M-r)^2} \bigg\{ a^2M \big[ -Mq+(2E^2+q)r\big] + r \big[ 4M^3q-M^2(5E^2+8q)r \\ 
\nonumber & \hspace{4.9cm}+M(4E^2+5q)r^2-(E^2+q)r^3 \big] \pm 2M \sqrt{f(r,q)} \bigg\}, 
\end{align}
where $f(r,q):=E^2r\big[-Mq+(E^2+q)r\big] \Delta(r)^2$ with  $-Mq+(E^2+q)r\leq 0$.

\end{prop}

\begin{proof}
Since $E\neq 0$ by Prop.  \ref{Q<0 condition},  we can divide the $r$-equation by $E^2$ to get

\begin{align*}
\bigg( \frac{\rho^2}{E}\bigg)^2(r')^2=& \bigg( 1+\frac{q}{E^2}\bigg)r^4-2M\frac{q}{E^2}r^3+\bigg[a^2\bigg( 1+\frac{q}{E^2}\bigg)-\Phi^2-\mathcal{Q}\bigg]r^2\\
& +2M\bigg[ \mathcal{Q}+(\Phi-a)^2\bigg]r-a^2\mathcal{Q}=:\mathcal{R}(r),
\end{align*}
where $\Phi:=L/E$ and $\mathcal{Q}:=Q/E^2$.  A geodesic has constant radial behaviour if and only if $\mathcal{R}(r)=0$ and $d\mathcal{R}(r)/dr=0$.  These two equations,  quadratic in $\Phi$,  can be solved for $\mathcal{Q}$ and $\Phi$ to get the two pairs \eqref{class r=const},  as analogously done in Prop.  $5.10$ of \cite{sanzeni2024non}.  (See also eq.  $(285)$ at p.  $363$ of \cite{Chandrasekhar}.) 

The condition $-Mq+(E^2+q)r\leq 0$ is required for the reality of \eqref{class r=const}.
\end{proof}

\begin{lem}\label{lemm B<0 for class +}
For the class $(\Phi_{+},  \mathcal{Q}_{+})$ of spherical timelike geodesics with $Q<0$ at radius $-M<r<0$,  the corresponding function $B\big(r,a,\Phi_{+}(r,q)\big)<0$ if and only if $q< \frac{E^2}{4M^2} (r^2+3Mr)$.
\end{lem}

\begin{proof}
The inequality $B\big(r,a,\Phi_{+}(r,q)\big)<0$ is equivalent to

\begin{align*}
 \frac{-2Mr \sqrt{f(r,q)}}{E^2}< r^2 (M+r)\Delta(r).
\end{align*}
Since $\Delta(r)>0$ in $r<0$,  $-M<r<0$,  the latter is equivalent to

\begin{align*}
E^2 r (r-M) [-4 M^2 q + E^2 (r^2 + 3 M r) ]  > 0,
\end{align*}

or equivalently

\begin{align*}
q< \frac{E^2}{4M^2} (r^2+3Mr)<0.
\end{align*}
\end{proof}

\begin{oss}
In the limit $q\to 0$,  the class $\big( \Phi_{-}, \mathcal{Q}_{-}\big)$ of spherical timelike geodesics reduces to the admissible class of spherical null geodesics $(21)$ of \cite{sanzeni2024non} while the class $(\Phi_{+},\mathcal{Q}_{+})$ reduces to the impossible class $(23)$ of \cite{sanzeni2024non}.  As seen at the beginning of \ref{case Q<0},  a necessary condition for the existence of the geodesic $\theta$-behaviour is $w=a^2-\hat{\Phi}^2-\hat{\mathcal{Q}}>0$.  One can sees that for the class $(\Phi_{+},\mathcal{Q}_{+})$,  $w=a^2-\Phi_{+}^2\frac{E^2}{E^2+q}-\mathcal{Q}_{+}\frac{E^2}{E^2+q}$ indeed can be positive for some negative $q<\frac{E^2}{4M^2} (r^2+3Mr)$ and some $-M<r<0$ and it is negative when $q=0$.
\end{oss}

\vspace{0.5cm}






We rule out closed spherical geodesics in the subcase  $E\neq 0$,  $Q<0$.  Consider a timelike geodesic $\gamma:I\rightarrow K^*$ with constants of motion $E\neq 0$,  $Q<0$,  non-constant coordinate functions $s\mapsto t(s),\theta(s),\phi(s)$ and constant negative radial coordinate $-M<r<0$ such that $B(r,a,\Phi)<0$.  The differential equation \eqref{t theta eq} has the form

\begin{align*}
\frac{dt}{d\theta}=F(\theta),
\end{align*}
for some function $F$.  The variation of the $t$-coordinate on a full $\theta$-oscillation is given by 

\begin{align*}
\Delta t= 2\int_{\theta_1}^{\theta_2}F(\theta) d\theta.
\end{align*}

\begin{oss}
Notice the factor $"2"$ in the last expression.  On a full $\theta$-oscillation,  we have 
\begin{align*}
\int_{\theta_1}^{\theta_2} F(\theta)d\theta +\int_{\theta_2}^{\theta_1} -F(\theta)d\theta=2\int_{\theta_1}^{\theta_2}F(\theta) d\theta.
\end{align*}
\end{oss}

Therefore the variation of the $t$-coordinate after $n$ $\theta$-oscillations is $n\Delta t$ because of the periodicity of the $\theta$-coordinate.  If the geodesic is closed,  $\Delta t=0$, otherwise the coordinate $t(s)$ cannot be periodic.  Hence it suffices to study what happens on a single $\theta$-oscillation.

\begin{oss}\label{same integrals}
A motion of the kind $\pi\geq \theta_4\geq\theta\geq\theta_3>\pi/2$ produces the same integrals since in this $\theta$-interval $\cos\theta<0$,  hence with the substitution $u=\cos^2\theta$  we have $d\theta=\frac{1}{2}\frac{du}{\sqrt{u}\sqrt{1-u}}$.  Therefore
\begin{align*}
\int_{\theta_1}^{\theta_2} \frac{d\theta}{\sqrt{\Theta(\theta)}}=\int_{\theta_3}^{\theta_4} \frac{d\theta}{\sqrt{\Theta(\theta)}},\hspace{1cm}  \int_{\theta_1}^{\theta_2} \frac{\cos^2\theta\;d\theta}{\sqrt{\Theta(\theta)}}=\int_{\theta_3}^{\theta_4} \frac{\cos^2\theta\;d\theta}{\sqrt{\Theta(\theta)}}.
\end{align*}
Hence $\Delta t$ is the same.
\end{oss}

So without any loss of generality,  we may consider a motion of the type $0\leq \theta_1\leq\theta\leq\theta_2<\pi/2$.  Then we can integrate \eqref{t theta eq} on a full oscillation to get

\begin{align}
\frac{\sqrt{E^2+q}}{E} \Delta t=2B(r,a,\Phi)\int_{\theta_1}^{\theta_2}\frac{d\theta}{\sqrt{\Theta(\theta)}}+2a^2 \int_{\theta_1}^{\theta_2}\frac{cos^2\theta\; d\theta}{\sqrt{\Theta(\theta)}}.
\end{align}

We now have to compute the following integrals

\begin{align*}
I_1:=&\int_{\theta_1}^{\theta_2} \frac{d\theta}{\sqrt{\Theta(\theta)}},\\
I_2:=&\int_{\theta_1}^{\theta_2} \frac{\cos^2\theta\;d\theta}{\sqrt{\Theta(\theta)}}.
\end{align*}

Let us start from the first integral:
\begin{align}
I_1=-\frac{1}{2}\int_{u_{+}}^{u_{-}}\frac{du}{\sqrt{u}\sqrt{\hat{\Theta}(u)}},
\end{align}
where we have used the substitution $u:=\cos^2\theta$,  hence $d\theta=-\frac{1}{2}\frac{du}{\sqrt{u}\sqrt{1-u}}$ since\\ $\sin\theta\geq 0$ and $\cos\theta>0$ if $\theta_1\leq \theta\leq\theta_2$.  Now we can use \eqref{theta di u} and the substitution\\ $u=:u_{-}+(u_{+}-u_{-})y^2$ adopted in \cite{Kapec_2019} to get

\begin{align*}
\nonumber   I_1=&\frac{1}{2}\int_{u_{-}}^{u_{+}}\frac{du}{\sqrt{u}\sqrt{a^2(u_{+}-u)(u-u_{-})}}\\ 
\nonumber =&\frac{1}{2|a|}\int_0^1 \frac{2(u_{+}-u_{-})ydy}{\sqrt{u_{-}+(u_{+}-u_{-})y^2}\sqrt{\big(u_{+}-u_{-}-(u_{+}-u_{-})y^2\big)(u_{+}-u_{-})y^2}}\\ 
\nonumber =& \frac{1}{|a|}\int_0^1 \frac{dy}{\sqrt{u_{-}+(u_{+}-u_{-})y^2}\sqrt{1-y^2}}\\ \nonumber
 \nonumber=& \frac{1}{|a|\sqrt{u_{-}}} \int_0^1 \frac{dy}{\sqrt{1-y^2}\sqrt{1-\big( 1-\frac{u_{+}}{u_{-}} \big) y^2}}.
\end{align*}
With the same substitutions,  we also get

\begin{align*}
\nonumber I_2=&-\frac{1}{2}\int_{u_{+}}^{u_{-}}\frac{udu}{\sqrt{u}\sqrt{\hat{\Theta}(u)}}\\
\nonumber =&\frac{1}{2} \int_{u_{-}}^{u_{+}}\frac{udu}{\sqrt{u}\sqrt{a^2(u_{+}-u)(u-u_{-})}}\\
\nonumber =& \frac{1}{|a|}\int_0^1 \frac{\sqrt{u_{-}+(u_{+}-u_{-})y^2}}{\sqrt{1-y^2}}dy\\
\nonumber =& \frac{\sqrt{u_{-}}}{|a|}\int_0^1 \frac{\sqrt{1-\big(1-\frac{u_{+}}{u_{-}} \big)y^2}}{\sqrt{1-y^2}}dy.\\
\end{align*}

Then with the definition of the elliptic integrals in Appendix \ref{elliptic appendix}  we have

\begin{align}
I_1=&\frac{1}{|a|\sqrt{u_{-}}} \mathcal{K}\bigg(1-\frac{u_{+}}{u_{-}} \bigg),\\
I_2=&\frac{\sqrt{u_{-}}}{|a|} \mathcal{E} \bigg( 1-\frac{u_{+}}{u_{-}} \bigg).
\end{align}

Hence,  we get

\begin{align}\label{formula delta t before last one}
\Delta t=\frac{E}{\sqrt{E^2+q}}\bigg[ \frac{2B(r,a,\Phi)}{|a|\sqrt{u_{-}}} \mathcal{K}\bigg(1-\frac{u_{+}}{u_{-}} \bigg) +  2|a|\sqrt{u_{-}} \mathcal{E} \bigg( 1-\frac{u_{+}}{u_{-}} \bigg)\bigg].
\end{align}

Note that,  since $u_{+}>u_{-}>0$,  we have $1-\frac{u_{+}}{u_{-}}<0$,  and hence $\mathcal{E}(1-u_{+}/u_{-})>\mathcal{K}(1-u_{+}/u_{-})>0$ (see Appendix \ref{elliptic appendix}).
However,  the prefactor of $\mathcal{E}$ does not dominate the opposite of the prefactor of $\mathcal{K}$ for every negative $r$,  as one may check substituting $\hat{\Phi}=\Phi_{\pm}E/\sqrt{E^2+q}$ and $\hat{\mathcal{Q}}=\mathcal{Q}_{\pm}E^2/(E^2+q)$ from \eqref{class r=const} into $u_{\pm}$ given by \eqref{formula for u pm} and $\Phi=\Phi_{\pm}$ into $B(r,a,\Phi)$.  

From now on set $x:=1-u_{+}/u_{-}$.  The elliptic integral $\mathcal{K}$ can be written as a hypergeometric function (see \ref{elliptic as hypergeom}): 

\begin{align*}
\mathcal{K}(x)=\frac{\pi}{2}F\bigg(\frac{1}{2},\frac{1}{2};1;x\bigg).
\end{align*}
Using the Pfaff transformation (see \ref{hypergeometric trick})

\begin{align}\label{hypergeo trick}
F\bigg( \alpha,\beta;\gamma;x\bigg)=(1-x)^{-\alpha}F\bigg( \alpha,\gamma-\beta;\gamma;\frac{x}{x-1} \bigg),
\end{align}
we can decrease the modulus of the prefactor in front of the elliptic integral $\mathcal{K}$:

\begin{align}\label{transformation k}
\mathcal{K}(x)=\frac{\sqrt{u_{-}}}{\sqrt{u_{+}}}\mathcal{K}\bigg(\frac{x}{x-1}\bigg).
\end{align}
Hence we get

\begin{align}\label{last espression for delta t}
\Delta t=\frac{E}{\sqrt{E^2+q}} \bigg[ 2|a|\sqrt{u_{-}}\mathcal{E}(x)+\frac{2B(r,a,\Phi)}{|a|\sqrt{u_{+}}}\mathcal{K}\bigg(\frac{x}{x-1}\bigg) \bigg].
\end{align}
Now we compare the elliptic integrals,  after the Pfaff transformation.  Since $x<0$,  we have

\begin{align}\label{inequality elliptic integrals}
\mathcal{E}(x)>\mathcal{K}\bigg(\frac{x}{x-1}\bigg)>0,
\end{align}
by Rmk.   \ref{remark integral estimate}.  Next we claim that the prefactors of the elliptic integrals in  \eqref{last espression for delta t} satisfy

\begin{align}\label{ineq prefactors}
2|a|\sqrt{u_{-}}>-\frac{2B(r,a,\Phi)}{|a|\sqrt{u_{+}}}.
\end{align}
Indeed,  both sides of the inequality are positive,  so we can square them and use that $u_{+}u_{-}=-\hat{\mathcal{Q}}/a^2=-\frac{\mathcal{Q}}{a^2}\frac{E^2}{E^2+q}$ by \eqref{theta di u} to get an equivalent inequality

\begin{align}\label{last inequality}
-a^2\mathcal{Q}\frac{E^2}{E^2+q}>B^2(r,a,\Phi).
\end{align}

\begin{prop}\label{last proposition}
For both classes \eqref{class r=const} of $r=\textrm{const}$ timelike geodesics with $\mathcal{Q}<0$ the inequality \eqref{last inequality} holds if $B(r,a,\Phi_{\pm})< 0$ and $-M<r<0$.
\end{prop}

\begin{proof}
Consider the class $(\Phi_{-},\mathcal{Q}_{-})$.  Then \eqref{last inequality} becomes 

\begin{align}\label{first inequality last prop}
& E^2 M (E^2 + q) (M - r)^2 r^2 \Delta(r) A(r,q) >0,  
\end{align}
with 

\begin{align*}
A(r,q):=  & - \Delta(r) \bigg\{ a^2E^2 \big[ -Mq + (2E^2+q)r\big] + r\big[ -4M^2q^2 +Mq(E^2+4q)r + E^2 (6E^2+7q)r^2 \big] \bigg\} \\
& +2a^2E^2 \sqrt{f(r,q)} +2r \big[3 E^2 r + 2 q (M + r)\big] \sqrt{f(r,q)} .
\end{align*}

Since $E^2+q>0$ by Prop.  \ref{Q<0 condition},  $\Delta(r)>0$ if $r<0$,   \eqref{first inequality last prop} is equivalent to 

\begin{align}\label{A<0}
A(r,q)>0.  
\end{align}

Since $-M<r<0$,  $q<0$,   we have 

\begin{align}\label{A(r,q)<G(r,q)}
A(r,q)  \geq  & - \Delta(r) \bigg\{ a^2E^2 \big[ -Mq + (2E^2+q)r\big] + r\big[ -4M^2q^2 +Mq(E^2+4q)r + E^2 (6E^2+7q)r^2 \big] \bigg\} \\
\nonumber =: & G(r,q).
\end{align}
We claim that $G(r,q)>0$,  hence \eqref{A<0} is satisfied.  Indeed,  first $ -Mq + (2E^2+q)r<-Mq + (E^2+q)r\leq 0$ by Prop.  \ref{prop spherical geodesics}.  Second we show that 

\begin{align*}
-4M^2q^2 +Mq(E^2+4q)r + E^2 (6E^2+7q)r^2 \geq 0.
\end{align*}
Since  $-Mq + (E^2+q)r\leq 0$,  we can respectively use $q^2(r-M)\geq -E^2qr$ and $r^2(E^2+q)\geq Mqr$ in the following

\begin{align*}
&-4M^2q^2 +Mq(E^2+4q)r + E^2 (6E^2+7q)r^2 \\
&= ME^2qr + 6E^4r^2 + 4Mq^2 (r-M) + 7E^2qr^2\\
& \geq ME^2qr + 6E^4r^2 -4ME^2qr + 7E^2qr^2\\
& = -3ME^2qr + 6E^4r^2 + 7E^2qr^2\\
& = -3ME^2qr + 6E^2r^2 (E^2+q) + E^2qr^2\\
& \geq -3ME^2qr + 6ME^2 qr + E^2qr^2\\
& = qE^2r (r+3M)\\
& >0,
\end{align*}
where in the last inequality we used that $q<0$,  $E\neq 0$ and $-M<r<0$.\\



Consider now the class $(\Phi_{+},\mathcal{Q}_{+})$.  From Lemma \ref{lemm B<0 for class +},  $B(r,a,\Phi_{+})<0$ implies that $q<\frac{E^2}{4M^2} (r^2+3Mr)$.  Prop.  \ref{prop spherical geodesics} implies that $q\geq -\frac{E^2r}{r-M}$.  Since $-M<r<0$,  we have

\begin{align*}
-\frac{E^2r}{r-M} < \frac{E^2}{4M^2} (r^2+3Mr)<0.
\end{align*}
Therefore we must show that \eqref{last inequality} holds for the class $(\Phi_{+},\mathcal{Q}_{+})$ when $r\in(-M,0)$ and

\begin{align*}
q\in\bigg[ -\frac{E^2r}{r-M} , \frac{E^2}{4M^2} (r^2+3Mr) \bigg).
\end{align*}
For this class,  \eqref{last inequality} is equivalent to

\begin{align*}
E^2 M (E^2 + q) (M - r)^2 r^2 \Delta(r) \bigg\{ -G(r,q)+2a^2E^2 \sqrt{f(r,q)} +2r \big[ 3E^2r +2q(M+r)\big] \sqrt{f(r,q)}\bigg\}<0.
\end{align*}
Since $E^2+q>0$ by Prop.  \ref{Q<0 condition},  $\Delta(r)>0$ if $r<0$,  $-M<r<0$  the latter is equivalent to

\begin{align*}
 -G(r,q)<-2a^2E^2 \sqrt{f(r,q)} -2r \big[ 3E^2r +2q(M+r)\big] \sqrt{f(r,q)}<0.
\end{align*} 
Now we square the latter,  reversing the sign of the inequality,  to get 

\begin{align} \label{last inequality class +}
a^4 E^4 - 2 a^2 E^2 r (-4 M q + E^2 r) + 
  r^2 \big[16 M^2 q^2 + 8 E^2 M q r - E^2 (15 E^2 + 16 q) r^2\big]>0.
\end{align}
We claim that  

\begin{align*}
&-4 M q + E^2 r >0,\\
&16 M^2 q^2 + 8 E^2 M q r - E^2 (15 E^2 + 16 q) r^2>0,
\end{align*}
hence \eqref{last inequality class +} holds.  

Indeed,  using that $q<\frac{E^2}{4M^2} (r^2+3Mr)$ and $-M<r<0$,  we have

\begin{align*}
-4 M q + E^2 r > -4M \frac{E^2}{4M^2} (r^2+3Mr) + E^2 r = -E^2 r \big( \frac{r}{M}+2\big)>0,
\end{align*}

\begin{align*}
16 M^2 q^2 + 8 E^2 M r q - E^2 (15 E^2 + 16 q) r^2 >& 16 M^2 \bigg[  \frac{E^2}{4M^2} (r^2+3Mr) \bigg]^2 +8E^2Mr \bigg[ \frac{E^2}{4M^2} (r^2+3Mr)  \bigg]\\
& - E^2 (15 E^2 + 16 q) r^2\\
=&\frac{E^2 r^2}{M^2} (E^2r^2 + 8E^2 Mr -16M^2q),
\end{align*}
which is positive if and only if

\begin{align*}
q< \frac{E^2r^2 + 8E^2Mr}{16M^2}.
\end{align*}
The last indeed holds because $-M<r<0$,  hence

\begin{align*}
\frac{E^2r^2 + 8E^2Mr}{16M^2} >  \frac{E^2}{4M^2} (r^2+3Mr).
\end{align*}

\end{proof}

Combining  \eqref{last espression for delta t},  \eqref{inequality elliptic integrals} and Prop.  \ref{last proposition}  we conclude that $\Delta t>0$ for all $-M<r<0$ such that $B(r,a,\Phi)< 0$,  which shows that the spherical timelike geodesics cannot be closed.

We have ruled out all the possibilities on Fig. \ref{figure steps of proof},  therefore there are no closed timelike geodesics in the Kerr-star spacetime.

\section{Conclusion}

We considered the Kerr-star spacetime,  namely the analytical extension of the Kerr spacetime over the horizons and in the negative radial region.  This spacetime contains closed timelike curves through every point below the inner horizon,  in the BL block III.  However we proved that the timelike geodesics cannot be closed.  Using simple geometrical arguments and the first integral differential equations,  we ruled out closed geodesics intersecting horizons,  those in the axis and most of the remaining ones.  It turned out that the most difficult geodesics to analyse were the spherical ones at negative radii with negative Carter constant.  We first excluded those with constant $\theta$-coordinate in Prop.  \ref{u-=u+ Q<0 prop}.  Then in Prop.  \ref{prop spherical geodesics} we found the two classes of spherical timelike geodesics parametrized by the constant $r$-coordinate and the Lorentzian energy $q$.  Hence we computed their variation $\Delta t$ of the $t$-coordinate on a full $\theta$-oscillation in eq.  \eqref{formula delta t before last one}.  Using the Pfaff transformation on the first elliptic integral in \eqref{formula delta t before last one},  we were able to compare the elliptic integrals in the resulting expression of $\Delta t$,  eq. \eqref{last espression for delta t}.  Finally we proved Prop.  \ref{last proposition} showing that $\Delta t$ is in fact positive for any spherical geodesic with negative Carter constant.\\

In this article  we proved the nonexistence of closed timelike geodesic in the Kerr-star spacetime (analytical extension of the slow Kerr spacetime over the horizons and in the negative radial region).  Combining Thm.  \ref{main theorem} with the result in \cite{sanzeni2024non},  it follows that despite the existence of causality violations,  the Kerr-star spacetime does not contain closed causal geodesics.

\vspace{0.5cm}

\normalsize {\thanks{ \noindent \textbf{Acknowledgements.} I would like to thank my PhD supervisors S.  Nemirovski and S.  Suhr for many fruitful discussions. }

\vspace{0.5cm}

\thanks{\noindent \textbf{Funding.} This research is funded by the Deutsche Forschungsgemeinschaft (DFG,  German Research Foundation) – Project-ID 281071066 – TRR 191.

\addtocontents{toc}{\setcounter{tocdepth}{0}}
\section*{Declarations}

\subsection*{Financial or Non-financial Intersts} The author has no relevant financial or non-financial interests to disclose.

\subsection*{Conflict of interest} The author has no competing interests to declare that are relevant to the content of this article.

\subsection*{Data availability statement}
Data sharing is not applicable to this article as no new data were created or analyzed in this study.

\appendix
\addcontentsline{toc}{chapter}{Appendices}
\addtocontents{toc}{\setcounter{tocdepth}{-1}}

\section{Elliptic integrals and hypergeometric functions}\label{elliptic appendix}

\begin{defn}

Let $\phi\in[-\pi/2,\pi/2]$.  The elliptic integral of the first kind is

\begin{align*}
\mathcal{F}(\phi|k):=\int_0^{\sin\phi}\frac{ds}{\sqrt{(1-s^2)(1-ks^2)}}.
\end{align*}

The complete $(\phi=\pi/2)$ elliptic integral of the first kind is

\begin{align*}
\mathcal{K}(k):=\mathcal{F}(\pi/2|k)=\int_0^1\frac{ds}{\sqrt{(1-s^2)(1-ks^2)}}.
\end{align*}

The elliptic integral of the second kind is

\begin{align*}
\mathcal{E}(\phi|k):=\int_0^{\sin\phi}\sqrt{\frac{1-ks^2}{1-s^2}}ds.
\end{align*}

The complete $(\phi=\pi/2)$ elliptic integral of the second kind is

\begin{align*}
\mathcal{E}(k):=\mathcal{E}(\pi/2|k)=\int_0^1\sqrt{\frac{1-ks^2}{1-s^2}}ds.
\end{align*}

We define also
\begin{align*}
\mathcal{D}(k):=\int_0^1\frac{s^2ds}{\sqrt{(1-s^2)(1-ks^2)}}=\frac{\mathcal{K}(k)-\mathcal{E}(k)}{k}=-2\frac{\partial \mathcal{E}(k)}{\partial k}.
\end{align*}

\end{defn}

\begin{oss}\label{remark integral estimate}
Let $0<z,s<1, \; x<0$.  
\begin{align*}
\sqrt{\frac{1-zs^2}{1-s^2}}>\frac{1}{\sqrt{1-s^2}\sqrt{1-xs^2}}\hspace*{0.5cm}\Longleftrightarrow\hspace*{0.5cm}(1-zs^2)(1-xs^2)>1 \hspace*{0.5cm}\Longrightarrow\hspace*{0.5cm}\mathcal{E}(x)>\mathcal{K}(z).
\end{align*}
If $z=x/(x-1)$,  it satisfies $0<z<1$ and we have 
\begin{align*}
(1-zs^2)(1-xs^2)>1 \hspace*{0.5cm}\Longleftrightarrow\hspace*{0.5cm} x+z<xzs^2\hspace*{0.5cm}\Longleftrightarrow\hspace*{0.5cm} 1>s^2,
\end{align*}
hence $\mathcal{E}(x)>\mathcal{K}(z)$.
\end{oss}

\begin{defn}[\cite{hypergeo}]
The hypergeometric function $F(\alpha,\beta;\gamma;x)$ is defined by the series 

\begin{align*}
\sum_{n=0}^{\infty} \frac{(\alpha)_n (\beta)_n}{(\gamma)_n n!}x^n,
\end{align*}
where $(\alpha)_n:=\alpha(\alpha+1)\cdot ... \, \cdot (\alpha+n-1)$ for $n>0$,  $(\alpha)_0\equiv 1$ (analogous for the others),  for $|x|<1$,  and by continuation elsewhere.
\end{defn}

\begin{prop}[\textit{Euler's integral representation},  see \cite{hypergeo}]\label{Euler integral rep}
If $\textrm{Re}\;\gamma>\textrm{Re}\;\beta>0$,  then

\begin{align*}
F(\alpha,\beta;\gamma;x)=\frac{\Gamma(\gamma)}{\Gamma(\beta)\Gamma(\gamma-\beta)}\int_0^1 t^{\beta-1} (1-t)^{\gamma-\beta-1}(1-xt)^{-\alpha}dt
\end{align*}
in the complex $x-$plane cut along the real axis from $1$ to $+\infty$,  where $\Gamma(x):=\int_0^{\infty} t^{x-1} e^{-t} dt$ is the Euler's gamma function.
\end{prop}

\begin{prop}[\cite{hypergeo}]\label{elliptic as hypergeom}
We can write the complete elliptic integral of the first kind as

\begin{align*}
\mathcal{K}(x)=\frac{\pi}{2} F\bigg(\frac{1}{2},\frac{1}{2};1;x\bigg).
\end{align*}
\end{prop}

\begin{proof}
Use the integral representation of the hypergometric function given in Prop. \ref{Euler integral rep},  the integral substitution $t=s^2$,  with $\Gamma(\frac{1}{2})=\sqrt{\pi}$,  $\Gamma(1)=1$.
\end{proof}

\begin{prop}[\textit{"Pfaff's formula",  see Theorem $2.2.5$ of \cite{hypergeo}}]\label{hypergeometric trick}

\begin{align*}
F(\alpha,\beta;\gamma;x)=(1-x)^{-\alpha} F\bigg(\alpha, \gamma-\beta;\gamma;\frac{x}{x-1}\bigg).
\end{align*}
\end{prop}

\begin{proof}
Use \ref{Euler integral rep} and the integral substitution $t=1-s$.
\end{proof}

\addcontentsline{toc}{chapter}{Conclusions}

\bibliography{bibliography}

\begin{thebibliography}{10}

\bibitem{hypergeo}
{\sc Andrews, G.~E., Askey, R., and Roy, R.}
\newblock {\em Special Functions}.
\newblock Encyclopedia of Mathematics and its Applications. Cambridge
  University Press, 1999.

\bibitem{Boyer-Lindquist_paper}
{\sc {Boyer}, R.~H., and {Lindquist}, R.~W.}
\newblock {Maximal analytic extension of the Kerr metric}.
\newblock {\em Journal of Mathematical Physics 8}, 2 (Feb. 1967), 265--281.

\bibitem{Boyer_Price_1965}
{\sc Boyer, R.~H., and Price, T.~G.}
\newblock An interpretation of the kerr metric in general relativity.
\newblock {\em Mathematical Proceedings of the Cambridge Philosophical Society
  61}, 2 (1965), 531–534.

\bibitem{Carter_1966_Axis}
{\sc Carter, B.}
\newblock {Complete analytic extension of the symmetry axis of Kerr's solution
  of Einstein's equations}.
\newblock {\em Phys. Rev. 141\/} (Jan 1966), 1242--1247.

\bibitem{Carter_causality}
{\sc Carter, B.}
\newblock {Global structure of the Kerr family of gravitational fields}.
\newblock {\em Phys. Rev. 174\/} (Oct 1968), 1559--1571.

\bibitem{Chandrasekhar}
{\sc Chandrasekhar, S.}
\newblock {\em {The mathematical theory of black holes}}.
\newblock 1983.

\bibitem{Chandr_Wright}
{\sc Chandrasekhar, S., and Wright, J.~P.}
\newblock {The geodesics in G\"odel's universe}.
\newblock {\em Proceedings of the National Academy of Sciences of the United
  States of America 47 3\/} (1961), 341--7.

\bibitem{Chrusc_singularity}
{\sc Chru\ifmmode~\acute{s}\else \'{s}\fi{}ciel, P., Maliborski, M., and Yunes,
  N.}
\newblock Structure of the singular ring in kerr-like metrics.
\newblock {\em Phys. Rev. D 101\/} (May 2020), 104048.

\bibitem{deFelice_1968}
{\sc de~Felice, F.}
\newblock {Equatorial geodesic motion in the gravitational field of a rotating
  source}.
\newblock {\em Nuovo Cim. B 57\/} (1968), 351.

\bibitem{Godel}
{\sc G\"odel, K.}
\newblock {An Example of a New Type of Cosmological Solutions of Einstein's
  Field Equations of Gravitation}.
\newblock {\em Rev. Mod. Phys. 21\/} (Jul 1949), 447--450.

\bibitem{Kapec_2019}
{\sc Kapec, D., and Lupsasca, A.}
\newblock Particle motion near high-spin black holes.
\newblock {\em Classical and Quantum Gravity 37}, 1 (dec 2019), 015006.

\bibitem{Kerr-paper}
{\sc Kerr, R.~P.}
\newblock {Gravitational field of a spinning mass as an example of
  algebraically special metrics}.
\newblock {\em Phys. Rev. Lett. 11\/} (Sep 1963), 237--238.

\bibitem{Kundt}
{\sc Kundt, V.~W.}
\newblock {Trägheitsbahnen in einem von Gödel angegebenen kosmologischen
  Modell}.
\newblock {\em Zeitschrift für Physik 145\/} (1956), 611--620.

\bibitem{Nolan_godel}
{\sc Nolan, B.}
\newblock {Causality violation without time-travel: closed lightlike paths in
  Gödel’s universe}.
\newblock {\em Classical and Quantum Gravity 37}, 8 (mar 2020), 085007.

\bibitem{KBH_book}
{\sc O'Neill, B.}
\newblock {\em {The geometry of Kerr black holes}}.
\newblock Ak Peters Series. Taylor \& Francis, 1995.

\bibitem{sanzeni2024non}
{\sc Sanzeni, G.}
\newblock {Non existence of closed and bounded null geodesics in Kerr
  spacetimes}.
\newblock {\em https://arxiv.org/abs/2308.09631v3\/} (2024).

\bibitem{Schw_paper}
{\sc {Schwarzschild}, K.}
\newblock {{\"U}ber das Gravitationsfeld eines Massenpunktes nach der
  Einsteinschen Theorie}.
\newblock {\em Sitzungsberichte der K{\"o}niglich Preussischen Akademie der
  Wissenschaften\/} (Jan. 1916), 189--196.

\bibitem{spherical_photon}
{\sc Teo, E.}
\newblock {Spherical photon orbits around a Kerr black hole}.
\newblock {\em General Relativity and Gravitation 35}, 11 (2003).

\bibitem{Walker_Penrose}
{\sc Walker, M., and Penrose, R.}
\newblock {On quadratic first integrals of the geodesic equations for type
  $\{22\}$ spacetimes}.
\newblock {\em Commun. Math. Phys. 18\/} (1970), 265--274.

\bibitem{Wilkins}
{\sc Wilkins, D.}
\newblock Bound geodesics in the kerr metric.
\newblock {\em Phys. Rev. D 5\/} (Feb 1972), 814--822.

\end{thebibliography}
\bibliographystyle{acm}

\end{document}